\documentclass[a4paper,UKenglish,cleveref, autoref, thm-restate ]{lipics-v2021}
\nolinenumbers
\usepackage{enumitem}
\usepackage{todonotes}
\usepackage{ifthen}
\usepackage[appendix=inline]{apxproof} %content in the main text (arxiv)
\newboolean{short}
\setboolean{short}{false}

\newtheoremrep{lemma}[theorem]{Lemma}
\newtheoremrep{claim}[theorem]{Claim}
\newcommand{\tw}{{\mathtt{tw}}}
\newcommand{\pw}{{\mathtt{pw}}}

\newcommand{\vc}{{\mathtt{vc}}}
\newcommand{\vi}{{\mathtt{vi}}}
\newcommand{\td}{{\mathtt{td}}}
\newcommand{\below}{\operatorname{below}}
\newcommand{\lca}{\operatorname{lca}}
\newcommand{\ECNFADNF}{\exists 3 \operatorname{CNF}\forall\operatorname{DNF}}

\newcommand{\ut}{{\tt ut}}   %utility
\newcommand{\PP}{{\mathcal P}}  %partition

\title{Core Stability in Additively Separable Hedonic Games of Low Treewidth} 

\author{Tesshu Hanaka}{Kyushu University, Fukuoka, Japan }{hanaka@inf.kyushu-u.ac.jp}{https://orcid.org/0000-0001-6943-856X}{Partially supported by JSPS KAKENHI Grant Numbers JP21H05852, JP21K17707, JP22H00513, JP23H04388.}%TODO mandatory, please use full name; only 1 author per \author macro; first two parameters are mandatory, other parameters can be empty. Please provide at least the name of the affiliation and the country. The full address is optional. Use additional curly braces to indicate the correct name splitting when the last name consists of multiple name parts.

\author{Noleen K\"{o}hler}{University of Leeds, Leeds, UK}{scsnk@leeds.ac.uk}{https://orcid.org/0000-0002-1023-6530}{Partially supported by ANR project ANR-18-CE40-0025-01 (ASSK).}

\author{Michael Lampis}{Universit\'{e} Paris-Dauphine, PSL University, CNRS UMR7243, LAMSADE, Paris, France}{michail.lampis@lamsade.dauphine.fr}{https://orcid.org/0000-0002-5791-0887}{Supported by ANR project ANR-21-CE48-0022 (S-EX-AP-PE-AL).}

\authorrunning{T. Hanaka, N. K\"{o}hler and M. Lampis} %TODO mandatory. First: Use abbreviated first/middle names. Second (only in severe cases): Use first author plus 'et al.'

\Copyright{Jane Open Access and Joan R. Public} %TODO mandatory, please use full first names. LIPIcs license is "CC-BY";  http://creativecommons.org/licenses/by/3.0/


 %TODO mandatory; please add comma-separated list of keywords

\category{} %optional, e.g. invited paper

\relatedversion{} %optional, e.g. full version hosted on arXiv, HAL, or other respository/website
\EventEditors{John Q. Open and Joan R. Access}
\EventNoEds{2}
\EventLongTitle{42nd Conference on Very Important Topics (CVIT 2016)}
\EventShortTitle{CVIT 2016}
\EventAcronym{CVIT}
\EventYear{2016}
\EventDate{December 24--27, 2016}
\EventLocation{Little Whinging, United Kingdom}
\EventLogo{}
\SeriesVolume{42}
\ArticleNo{23}
\begin{document}

\maketitle

\begin{abstract}
Additively Separable Hedonic Game (ASHG) are coalition-formation games where we are given a graph whose vertices represent $n$ selfish agents and the weight of each edge $uv$ denotes how much agent $u$ gains (or loses) when she is placed in the same coalition as agent $v$. We revisit the computational complexity of the well-known notion of core stability of ASHGs, where the goal is to construct a partition of the agents into coalitions such that no group of agents would prefer to diverge from the given partition and form a new (blocking) coalition. Since both finding a core stable partition and verifying that a given partition is core stable are intractable problems ($\Sigma_2^p$-complete and coNP-complete respectively) we study their complexity from the point of view of structural parameterized complexity, using standard graph-theoretic parameters, such as treewidth.

We begin by presenting several results on \textsc{Core Stability Verification} (CSV) indicating that this is an unusually intractable problem, even in very restricted cases: \textsc{CSV} remains coNP-complete on graphs of vertex cover 2; \textsc{CSV} is coW[1]-hard parameterized by vertex integrity when edge weights are polynomially bounded; and \textsc{CSV} is coW[1]-hard parameterized by tree-depth even if all weights are from $\{-1,1\}$. We complement these results with essentially matching algorithms and one of the rare tractability results we present is that \textsc{CSV} is FPT parameterized by the treewidth $\tw$ plus the maximum degree $\Delta$ (improving a previous algorithm's dependence from $2^{O(\tw\Delta^2)}$ to $2^{O(\tw\Delta)}$).

We then move on to study \textsc{Core Stability} (CS), which one would naturally expect to be even harder than \textsc{CSV}. We confirm this intuition by showing that \textsc{CS} is $\Sigma_2^p$-complete even on graphs of bounded vertex cover. On the positive side, we consider the parameterization by $\tw+\Delta$ and improve the known algorithm (which was based on Courcelle's theorem) via an explicit algorithm utilizing a reduction to $\exists\forall$-\textsc{SAT}. The running time of our algorithm is, unfortunately, double-exponential in $\tw+\Delta$. However, we show that this is likely to be optimal, as the existence of an algorithm with less than double-exponential dependence on $\tw$ (even for bounded-degree instances) would contradict the Exponential Time Hypothesis (ETH).

Finally, we consider another natural parameter for these problems: the size $k$ of the considered blocking (diverging) coalitions. Fixing $k$ to be constant does lower the complexities of $\textsc{CSV}$ and $\textsc{CS}$ to P and NP-complete respectively, since one can then consider all coalitions of size $k$ in polynomial time. Unfortunately, we show that a more efficient algorithm is unlikely to exist, as \textsc{CSV} is coW[1]-hard parameterized by $k$ (even on unweighted graphs), while \textsc{CS} is NP-complete for all $k\ge 3$ (even on graphs of bounded degree with bounded edge weights).
\end{abstract}

%\newpage

\section{Introduction}

\emph{Coalition formation games} model situations where a group of selfish agents
need to be partitioned into teams (coalitions) in such a way that takes into
account their preferences. Because such games capture a vast array of
interesting scenarios, they have been a subject of intense study in
computational social choice and the social sciences at large.  One particularly
interesting and natural special case of such games is when the preferences of
each agent only depend on the other agents that she is placed together with in
the same coalition (and not on the placement of agents on other coalitions).
Such games are referred to in the literature as \emph{hedonic games} and have
also attracted much interest from the computer science perspective
(\cite{AloisioFV20,AzizBBHOP19,BarrotOSY19,BarrotY19,BoehmerE20,0001BW21,BullingerK21,FanelliMM21,IgarashiOSY19,OhtaBISY17,SliwinskiZ17}),
thanks in part to their numerous applications in, for example, social network
analysis \cite{Olsen09}, scheduling group activities \cite{DarmannEKLSW18}, and
allocating tasks to wireless agents \cite{SaadHBDH11}. For more information we
refer the reader to \cite{Cechlarova16} and the relevant chapters of standard
computational social choice texbooks \cite{AzizS16}.

Hedonic games are extremely general. Unforunately, this generality renders them hard to study from the computer science perspective -- indeed,
even listing the preferences of all $n$ agents takes space exponential in $n$
as the naïve approach would give the ordering of each agent over all coalitions. This motivates the study of natural restrictions
of hedonic games. In this paper we focus on one of the most natural such
restrictions: \emph{Additively Separable Hedonic Games} (ASHGs), where
the input is an edge-weighted graph,  vertices represent the agents, and
the weight of the edge $uv$ denotes the utility that agent $u$ derives from
being in the same coalition as $v$. The utility of an agent $u$ in a coalition
$C$ can then be succinctly encoded as the sum of the weights of edges incident
on $u$ with their other endpoint in $C$.

In any situation where agents behave selfishly, it becomes
critical to look for \emph{stable} outcomes, that is, outcomes which the agents
are likely to accept, based on their preferences. In the context of ASHGs, the
question then becomes: given an edge-weighted graph $G$ representing the
agents' preferences, can we find a stable partition of the agents into
coalitions (possibly also optimizing some other social welfare goal)? The
computational complexity of such questions has been amply studied
(\cite{AzizBS13,Ballester04,ElkindFF20,FlamminiKMZ21,HanakaKMO19,Olsen09,OlsenBT12,SungD10})
and several natural notions of stability have been proposed. In this paper we
revisit the computational complexity of one of the most well-studied such
notions, which is called \emph{core stability}. Intuitively, a partition of $n$
agents is called \emph{core stable}, if it is stable enough to dissuade not
only individual diverging behavior but even divergence by groups of agents.
More formally, given a partition $\PP$ of the agents, a \emph{blocking
	coalition} is a set of agents $X$ such that all $v\in X$ have strictly higher
utility in $X$ than in the initial partition $\PP$. Hence, if a blocking
coalition $X$ exists, the initial partition is \emph{unstable}, because the
agents of $X$ would prefer to form a new coalition. A partition
is then called core stable if no blocking coalition (of any size) exists.
Notice that core stability is a very strong (and hence very desirable) notion
of stability, compared with simpler notions, such as Nash stability (which only
precludes divergence by a single agent).

Attractive though it may be from the game theory point of view, the notion of
core stability presents some serious drawbacks from the point of view of
computational complexity. In particular, deciding if an ASHG admits a core
stable outcome is not just NP-hard, but in fact $\Sigma_2^p$-complete, that is,
complete for the second level of the polynomial hierarchy \cite{Woeginger13},
even if the input graph is undirected, has bounded degree, and edge weights are
bounded by a constant \cite{Peters17}. Compared to simpler notions of
stability, such as Nash stability (which is ``only'' NP-complete
\cite{GairingS19}), core stability is therefore highly intractable, and this
strongly motivates the search for a better understanding of what the source of
this intractability is and for ways to deal with it. The focus of this paper is
on using notions of graph structure from parameterized complexity to achieve a
more fine-grained understanding of the complexity of this problem. Throughout
the paper we will concentrate on the case where agent preferences are
\emph{symmetric}, that is, the given graphs are undirected. Since most of our
results are negative, this (natural) restriction only renders them stronger.

\subparagraph*{Our results} In this paper we present several results that improve
and clarify the state of the art on the complexity of finding core stable
outcomes in ASHGs. We study two closely related problems: \textsc{Core
	Stability (CS)} and \textsc{Core Stability Verification (CSV)}, which
correspond to deciding if a core stable partition exists and deciding if a
given partition is indeed core stable respectively. Intuitively, the reason
\textsc{CS} is complete for the second level of the polynomial hierarchy (and
not just NP-complete) is that \textsc{CSV} is also known to be intractable
(coNP-complete \cite{chen2023hedonic,SungD07}). Our high-level aim is to
understand which parts of the combinatorial structure of the input are
responsible for the complexity of these two problems. In order to quantify the
input structure we will use standard structural tools from the toolbox of
parameterized complexity, such as the notions of treewidth and related
parameters\footnote{Throughout the paper we assume the reader is familiar with
	the basics of parameterized complexity, as given for example in
	\cite{CyganFKLMPPS15}}. 

We begin our investigation with \textsc{CSV} and ask the question which
restrictions on the input are likely to render the problem tractable (or
conversely, what are the sources of the problem's intractability). We identify
two possible culprits: the problem could become easy if we either impose
restrictions on the graph structure, for example by requiring that the input be
of low treewidth or degree, or if we impose restrictions on the allowed edge
weights. Our results indicate that these two sources of intractability interact
in non-trivial ways:

\begin{itemize}
	
	\item If we place absolutely no restrictions on the allowed weights,
	\textsc{CSV} remains hard even on severely restricted instances, that is,
	graphs of vertex cover $2$ (Theorem \ref{thm:csv:vc:weak}). We find this rather
	surprising, as this class of graphs (which are essentially stars with one
	additional vertex) is rarely general enough to render problems intractable.
	
	\item One may be tempted to interpret the previous result as an artifact of the
	exponentially large weights we allow in the input. However, we show that even
	if we place the restriction that weights are polynomially bounded in the input
	size, \textsc{CSV} still remains quite hard from the parameterized perspective,
	and more precisely coW[1]-hard parameterized by vertex integrity (Theorem
	\ref{thm:csv:vi}).  Recall that graphs with small vertex integrity are graphs
	where there exists a small separator whose removal breaks down the graph into
	components of bounded size, so this parameterization is again rather
	restrictive and usually easily renders most problems almost as tractable as
	parameterizing by vertex cover \cite{GimaHKKO22,LampisM21}.
	
	\item Finally, we show that even if we insist on weights only being selected
	from the set $\{-1,1\}$, \textsc{CSV} is coW[1]-hard parameterized by tree-depth
	(Theorem \ref{thm:csv:td}).
	
\end{itemize}

Taken together these results show that \textsc{CSV} is an unusually intractable
problem where hardness comes from a combination of two factors: the complexity
of dealing with the edge weights and the complexity of dealing with the
graph-theoretic structure of the input. We complement the above with several
algorithms that paint a clearer picture of the complexity of \textsc{CSV}
showing that: (i) \textsc{CSV} is polynomial-time solvable on trees (Theorem
\ref{thm:csv:trees}), hence Theorem \ref{thm:csv:vc:weak} cannot be extended to
graphs of vertex cover $1$ (ii) \textsc{CSV} is FPT parameterized by vertex
integrity plus the maximum edge weight (Theorem \ref{thm:csv:FPT:vi}), so the
hardness result of Theorem \ref{thm:csv:td} cannot be extended to vertex
integrity (iii) Theorem \ref{thm:csv:td} is matched by an XP algorithm
parameterized by treewidth with parameter dependence $(\Delta
w_{\max})^{O(\tw)}$, that is, an XP algorithm when weights are polynomially bounded
(Theorem \ref{thm:tw:XP}) (iv) the former algorithm can be improved to an FPT
running time (even for unbounded weights) if we parameterize by $\tw+\Delta$
(this was already observed by Peters \cite{Peters16a}, who gave an algorithm
with dependence $2^{O(\Delta^2\tw)}$, but we improve this complexity to
$2^{O(\Delta\tw)}$ in  Theorem \ref{thm:tw+d:FPT}).

The results above paint a comprehensive and rather negative picture on the
complexity of \textsc{CSV}, which seems to imply that our main problem, that
is, \emph{finding} core-stable partitions, is likely to be even more
intractable. We confirm this intuition by showing that \textsc{CS} remains
$\Sigma_2^p$-complete even on graphs of bounded vertex cover (Theorem
\ref{thm:csfSigma2}).  One encouraging piece of news, however, is that we did
manage to obtain an FPT algorithm when \textsc{CSV} is parameterized by
$\tw+\Delta$, so this seems like a case worth considering for \textsc{CS}.
Indeed, Peters \cite{Peters16a} already showed that \textsc{CS} is FPT for this
parameterization, without, however, giving an explicit algorithm (his argument
was based on Courcelle's theorem). We improve upon this by giving an explicit
algorithm whose dependence is \emph{double-exponential} on $\tw+\Delta$, using
the technique of reducing to $\exists\forall$-\textsc{SAT} advocated in
\cite{LampisMM18} (Theorem \ref{thm:doubleExpAlgoCS}).  Despite fixed-parameter
tractability, it is fair to say that the running time of our algorithm is quite
disappointing.  Our main contribution in this part is to show that this is,
unfortunately, likely to be optimal: even for instances of bounded degree, the
existence of an algorithm with better than double-exponential dependence on
treewidth would violate the ETH (Theorem \ref{thm:ETHlowerBoundCS}). This shows
another aspect where core-stability is significantly harder than Nash
stability, which has ``just'' slightly super-exponential in $\tw+\Delta$
\cite{HanakaL22}. Note that the phenomenon that problems complete for the
second level of the polynomial hierarchy tend to have double-exponential
complexity in treewidth has been observed before
\cite{abs-2307-08149,LampisM17,MarxM16}

Finally, we conclude our paper by considering one last relevant parameter: the
size of the allowed blocking coalition. We say that a partition is $k$-core
stable if no blocking coalition of size at most $k$ exists. The concept of $k$-core stability was first proposed in \cite{FanelliMM21}. For small values of
$k$ this is a natural variation of the problem, which could potentially render
it more tractable -- indeed, for $k$ fixed, \textsc{CSV} is trivially in P and
\textsc{CS} is trivially in NP. Unfortunately, we show that not much more is
gained from these parameterizations: \textsc{CSV} is coW[1]-hard parameterized
by $k$ (even on unweighted graphs); while $k$-\textsc{CS} is NP-complete for
all fixed $k\ge 3$, even on graphs of bounded maximum degree and with bounded
weights.

\section{Preliminaries}

Throughout the paper we use standard graph-theoretic notation and focus on
undirected graphs. An Additively Separable Hedonic Game (ASHG) is represented
by a graph $G=(V,E)$, where vertices of $V$ represent the agents, and a weight
function $w:E\to \mathbb{Z}$.  A partition $\PP$ of $V$ is a collection of
disjoint subsets of $V$ whose union includes all of $V$. We will call the sets
of such a partition \emph{coalitions}. Slightly abusing notation, we will
write, for $u\in V$,  $\PP(u)$ to denote the set of $\PP$ that contains $u$.
The utility of an agent $u\in X$ in a set $X\subseteq V$ is defined as
$\ut(X,u) = \sum_{v\in X} w(uv)$, while the utility of $u$ in a partition $\PP$
is defined as $\ut_{\PP}(u) = \ut(\PP(u),u) = \sum_{v\in \PP(u)} w(uv)$. Even
though we defined $w$ as a function to the integers, we will sometimes allow
rational edge weights, but with denominators sufficiently small that it will
always be easy to obtain an equivalent integer instance by multiplying all
weights by an appropriate integer. We use $w_{\max}$ to denote the maximum
\emph{absolute} weight of a given ASHG instance. Unless otherwise stated, we
assume that $w$ is given to us encoded in binary (and hence $w_{\max}$ may have
value exponential in the input size).

We are chiefly interested in the following notion of stability.

\begin{definition}[Core stability]\label{def:core} A partition $\PP$ of an ASHG
	$(G,w)$ is \emph{core stable}, if there exists no $X\subseteq V(G)$ such that
	for all $u\in X$ we have $\ut(X,u)>\ut_{\PP}(u)$. \end{definition} 

If the set $X$ mentioned \cref{def:core} does exist, then we say that $\PP$ is
unstable and that $X$ is a \emph{blocking coalition}. For fixed integer values
of $k$, we will also study the notion of $k$-Core Stability: a partition is
$k$-core stable if no blocking coalition of size at most $k$ exists.

The two computational problems we are interested in are \textsc{Core Stability}
(\textsc{CS}) and \textsc{Core Stability Verification} (\textsc{CSV}). In the
former problem we are given as input an ASHG and are asked if there exists a
core stable partition; in the latter we also given a specific partition $\PP$
and are asked if $\PP$ is core stable.

We say that a partition $\mathcal{P}$ of $V(G)$ is \emph{connected} if
$G[P]$ is connected for every $P\in \mathcal{P}$. Notice that for both
\textsc{CSV} and \textsc{CS} we may assume that the partition $\PP$ we seek or
we are given is connected, as replacing a disconnected coalition $P\in\PP$ with
a coalition for each of its components does not change $\ut_{\PP}(u)$ for any
$u\in V$ and hence does not affect stability.

\subsection{Graph parameters and Parameterized Complexity}

We assume the reader is familiar with the basics of parameterized complexity,
such as the classes FPT and W[1], as given for example in
\cite{CyganFKLMPPS15}. One particularity is that \textsc{CSV} is a
coNP-complete problem, that is, a problem for which one can easily verify No
certificates (blocking coalitions). As a consequence, some of our results will
give coW[1]-hardness (rather than W[1]-hardness) for \textsc{CSV}, by which we
mean that the complement of any problem in W[1] can be fpt-reduced to the
problem at hand. We will say that a problem is \emph{weakly} hard for a class,
if the reduction we present uses exponentially large weights (which therefore
need to be encoded in binary). Conversely, if the reduction is valid even when
weights are encoded in unary (and are therefore polynomially bounded) we say
that the problem is \emph{strongly} hard.

We assume that the reader is also familiar with standard structural graph
parameters. The parameters we will focus on are treewidth ($\tw$),
pathwidth ($\pw$), tree-depth ($\td$), vertex integrity ($\vi$), and vertex
cover ($\vc$). For the definitions of treewidth and pathwidth, as well as the
corresponding (nice) decompositions we refer the reader to
\cite{CyganFKLMPPS15}. The vertex integrity $\vi(G)$ of a graph $G$ is defined
as the minimum $k$ such that there exists a set $S\subseteq V(G)$ (called a
$\vi(k)$-set) such that the largest component of $G-S$ has order at most
$k-|S|$.  The tree-depth of a graph $G$ is defined inductively as follows: an
isolared vertex has tree-depth $1$; the tree-depth of a disconnected graph is
the maximum of the tree-depth of its components; the tree-depth of a connected
graph $G$ is defined as $\min_{v\in V(G)} \td(G-v)+1$. The vertex cover of $G$
is the size of the smallest set of vertices of $G$ that intersects all edges.

It is well known that for all graphs $G$ we have $\tw(G)\le \pw(G)\le \td(G)\le
\vi(G) \le \vc(G)+1$. In terms of parameterized complexity these parameters
therefore form a hierarchy: if a problem is FPT for a smaller parameter, then
it is FPT for the larger ones and conversely if a problem is intractable for a
large parameter, then it is intractable for a smaller one. We therefore
say that larger parameters are more restrictive, with vertex cover being the
most restrictive parameter we consider. We use $\Delta(G)$ to denote the
maximum degree of a graph $G$. We omit $G$ from notation, if it is clear from context.

\section{Core Stability Verification}

In this section we study the complexity of \textsc{Core Stability Verification}
(\textsc{CSV}). What we discover is that this is an unusually intractable
problem, even for quite restricted parameterizations. In particular, we present
the three following hardness results:

\begin{itemize}
	
	\item \textsc{CSV} is weakly coNP-complete on graphs of vertex cover number 2.
	(\cref{thm:csv:vc:weak})
	
	\item \textsc{CSV} is strongly coW[1]-hard parameterized by vertex integrity.
	(\cref{thm:csv:vi})
	
	\item \textsc{CSV} is coW[1]-hard parameterized by tree-depth, even if all
	weights are in $\{-1,1\}$. (\cref{thm:csv:td}) 
	
\end{itemize}

Our results indicate that, even though some of the problem's intractability can
be attributed to the edge weights, a large part of its complexity is due to the
graph-theoretic structure of the input, and the problem is intractable for
tree-depth (and hence for pathwidth and treewidth) even for very small weights.
These hardness results are complemented by several algorithms, which show that
the hardness results are essentially tight. In particular:

\begin{itemize}
	
	\item In \cref{thm:csv:trees} we show that \textsc{CSV} is in P for trees,
	therefore, \cref{thm:csv:vc:weak} cannot apply to graphs of vertex cover 1. 
	
	\item In \cref{thm:csv:FPT:vc} we show that \textsc{CSV} is in XP parameterized
	by vertex cover when weights are polynomially bounded. This implies that
	\cref{thm:csv:vc:weak} cannot be improved to give strong coNP-completeness.
	
	\item In \cref{thm:csv:FPT:vi} we show that \textsc{CSV} is FPT parameterized
	by $\vi+w_{\max}$, so \cref{thm:csv:td} cannot be extended to \textsc{CSV}
	parameterized by vertex integrity.

	\item \cref{thm:csv:td} is matched by an XP algorithm parameterized by
	treewidth with parameter dependence $(\Delta w_{\max})^{O(\tw)}$, i.e. an XP
	algorithm when weights are polynomially bounded (\cref{thm:tw:XP}).
	
	\item The former algorithm can be improved to an FPT running time (even for
	unbounded weights) if we parameterize by $\tw+\Delta$. This was already
	observed by Peters \cite{Peters16a}, who gave an algorithm with dependence
	$2^{O(\Delta^2\tw)}$, but we improve this complexity to $2^{O(\Delta\tw)}$
	(\cref{thm:tw+d:FPT}).
	
\end{itemize}

\subsection{Hardness Results}
We first prove the following three hardness results.
\begin{theorem}\label{thm:csv:vc:weak}
	\textsc{Core Stability Verification} is weakly coNP-complete on graphs of vertex cover number 2.
\end{theorem}
\begin{figure}
	\centering
	\includegraphics{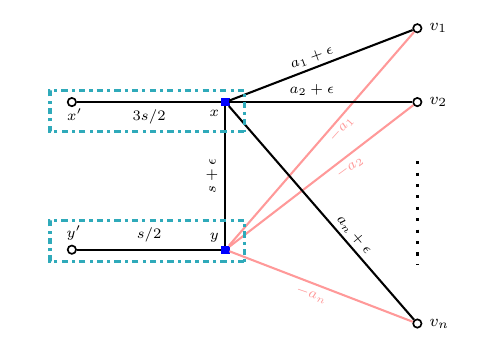}
	\caption{The graph $(G,w)$ constructed in the proof of \cref{thm:csv:FPT:vc}. The vertex cover is marked by blue squares. The initial partition is given by dot-dashed boxes, sngletons are omitted.}
	\label{fig:csv:vc:weak}
\end{figure}

\begin{proof}
	We give a reduction from \textsc{Partition}. Given a set of positive integers
	$A=\{a_1, \ldots, a_n\}$, the \textsc{Partition} problem asks whether there
	exists a subset $A'$ of $A$ such that $\sum_{a\in A'}a = s/2$ where $s=\sum_{a\in
		A} a$. This problem is well-known to be weakly NP-complete \cite{GareyJ79}.
	
	We construct an instance of \textsc{CSV}. First, we
	create $n$ vertices $v_{1},\ldots,v_{n}$ corresponding to $a_1,\ldots,
	a_n\in A$ and three vertices $x,y,x',y'$. Then we add edges $v_{i}x$ of
	weight $a_i+\epsilon$, $v_{i}y$ of weight $-a_i$, $xx'$ of
	weight $3s/2$, $yy'$ of
	weight $s/2$, and $xy$ of weight $s+\epsilon$. Here, without loss of
	generality, let $\epsilon$ be an integer sufficiently smaller than $\min_i
	a_i$; this can be achieved for example by multiplying all elements of $A$ (and
	$s$) by $n$, and setting $\epsilon=1$.  Let $(G,w)$ be the constructed graph
	(see Figure \ref{fig:csv:vc:weak}).  A coalition structure $\PP$ to verify
	consists of $\{x,x'\}$, $\{y,y'\}$ and singletons of other vertices. Also note that $\{x,y\}$ is a vertex cover of $G$.
	
	If there exists  $A'\subseteq A$ such that $\sum_{a\in A'}a = s/2$, then the coalition $X=\{v_i: a_i\in A'\}\cup \{x,y\}$ blocks $\PP$. To see this, observe that the utility of each vertex in $X$ increases by $\epsilon$ and thus $X$ is a blocking coalition of $\PP$.
	
	Conversely, suppose that there exists a blocking coalition $X$ of $\PP$.
	Clearly, $X$ contains neither $x'$ nor $y'$. If $X$ does not contain $x$, no vertex can have positive utility in $X$. Thus, they also do not join $X$.
	Consequently, $X=\{y\}$  holds, but this contradicts that $X$ is  a blocking coalition. 
	Thus, $X$ must contain $x$. To increase the utility
	of $x$, $y$ must be contained in $X$. In particular, if $y$ was not contained
	in $X$, then the utility of $x$ would be at most
	$s+n\epsilon < 3s/2$.
	%!!!\todo{PROBLEM: This does not appear to give a contradiction. We need to increase $xy$ to $3B$ and $yz$ to $2B+\epsilon$?} 
	Since $y$ must
	have utility more than $s/2$ in $X$, it holds that $\sum_{v_i\in X} a_i \le s/2$.
	Finally, as the utility of $x$ must increase by more than $s/2$ and $\epsilon$ is
	sufficiently smaller than $\min_i a_i$, $X$ satisfies that $\sum_{v_i\in X} a_i
	=s/2$, which implies that there exists a subset $A'$ of $A$ such that $\sum_{a\in
		A'}a = s/2$.  \end{proof}

We use a similar but more involved construction  to reduce \textsc{Bin Packing} to \textsc{CSV}.

\begin{theorem}\label{thm:csv:vi} \ifthenelse{\boolean{short}}{\textup{($\star$)}}{}
	\textsc{Core Stability Verification} is coW[1]-hard parameterized by vertex
	integrity %(plus cluster deletion number) 
	even if all weights are bounded by a polynomial in the input size.
\end{theorem}

\begin{toappendix}
	\begin{proof}[Proof\ifthenelse{\boolean{short}}{ of \cref{thm:csv:vi}}{}]
		\begin{figure}
			\centering
			\centerline{\includegraphics{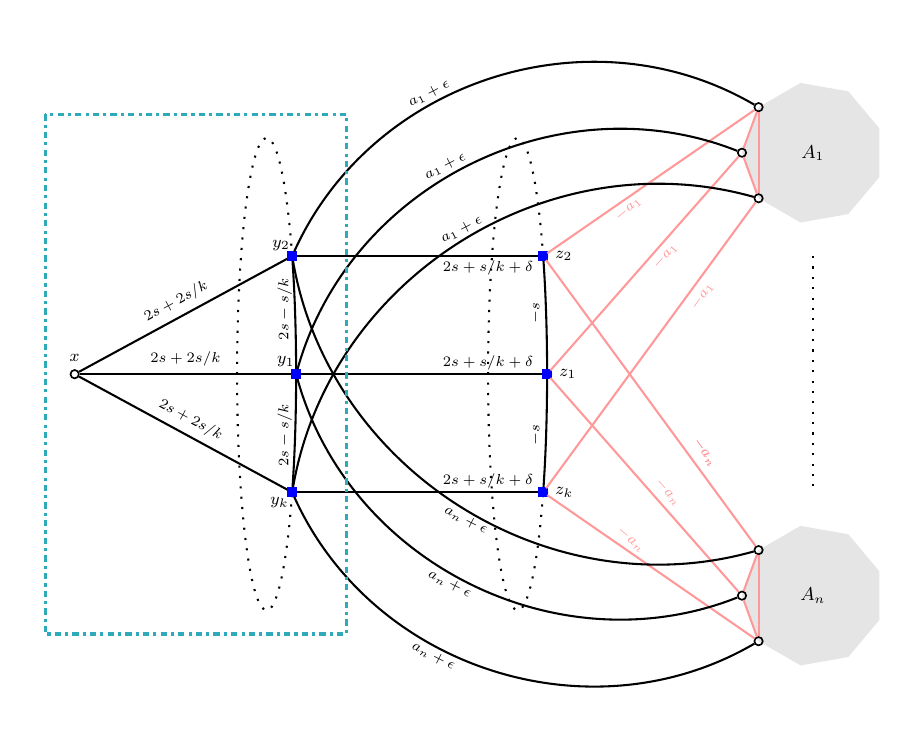}}
			\caption{The graph $(G,w)$ constructed in the proof of Theorem~\ref{thm:csv:vi}. The coalition structure $\mathcal{P}$ is indicated by the dash-dotted box (singletons are omitted). Furthermore, the vertices of the deletion set are represented by blue squares and  edges without label have weight $\rho$.}
			\label{fig:csv:vi}
		\end{figure}
		We give a reduction from \textsc{Bin Packing}. Given a set of positive integers
		$A=\{a_1, \ldots, a_n\}$ and an integer $k$, the \textsc{Bin Packing} problem
		asks whether there exists a partition $\mathcal{Q}=(Q_1,\dots, Q_k)$ of $A$
		such that $\sum_{a\in Q_i}a = s/k$ for every $i\in[k]$ where $s:=\sum_{a\in A}
		a$. This problem is known to be W[1]-hard parameterized by $k$, even if all
		integers are given in unary (that is, if weights are polynomially bounded in
		$n$) \cite{JansenKMS13}.
		
		We construct an instance of \textsc{CSV} as follows. First we construct a weighted graph $(G,w)$. We create a vertex $x$ and two cycles $(y_1,\dots,y_k)$ and $(z_1,\dots,z_k)$. Every edge of the cycle $(y_1,\dots, y_k)$ has weight $2s-s/k$ while every edge of the cycle $(z_1,\dots,z_k)$ has weight $-s$. Now, we create an edges $xy_i$ of weight $2s+2s/k$ and an edge $y_iz_i$ of weight $2s+s/k+\delta$ for every $i\in [k]$ where $\delta:=\frac{1}{4}$. For every $a_i\in A$ we create a clique $A_i$ consisting of $k$ vertices $v_i^1,\dots,v_i^k$ and edges of weight $\rho$ where $\rho$ is an integer smaller than $-(a_i+\epsilon)$. Finally, we create edges $v_i^jy_j$ of weight $a_i+\epsilon$ and edges $v_i^jz_j$ of weight $-a_i$ for every $i\in [n]$, $j\in [k]$ where $\epsilon:=\frac{1}{2n}$. Note that we can easily make all weights integer by multiplying all weights by $2n$.
		The coalition structure $\mathcal{P}$ we want to verify consists of the set $\{x,y_1,\dots,y_k\}$ and singletons for every vertex $v\in V(G)\setminus \{x,y_1,\dots,y_k\}$. We now argue that $(A,k)$ is a YES-instance of \textsc{Bin Packing} if and only if $((G,w), \mathcal{P})$ is a NO-instance of \textsc{CSV}.\\
		
		First assume that $(A,k)$ is a YES-instance of \textsc{Bin Packing} and $\mathcal{Q}=(Q_1,\dots, Q_k)$ is a partition such that $\sum_{a\in Q_i}a = s/k$ for every $i\in[k]$. We now claim that the set $X:=\{y_1,\dots,y_k,z_1,\dots,z_k\}\cup \bigcup_{j\in [k]}\{v_i^j: a_i\in Q_j\}$ is a blocking coalition for $\mathcal{P}$. To verify observe that the utility of every vertex $v_i^j\in X$ increases by $\epsilon$. Furthermore, the utility of $y_i$ improved by $\epsilon\cdot |Q_i|$ and the utility of $z_i$ by $\delta$ and hence $X$ is a blocking coalition of $\mathcal{P}$. \\
		
		Now assume that $(G,\mathcal{P})$ is a NO-instance of \textsc{CSV} and $X$ is a blocking coalition of $\mathcal{P}$. First note that $x$ cannot be in $X$ as its utility in $\mathcal{P}$ is maximum. Furthermore, $X$ must contain some element from $\{y_1,\dots,y_k\}$ as every edge weight in the graph $G[\{z_1,\dots,z_k\}\cup \{v_i^j:i\in [n],j\in [k]\}]$ is negative. Now observe that the utility of $y_i$ in the coalition $V(G)\setminus \{x\}$ is $7s-s/k+\delta+\epsilon n< 7s-s/k+1$ for every $i\in [k]$. Since the utility of $y_i$ is $6s$ in $\mathcal{P}$ this implies that if $y_i\in X$ then every neighbor $u$ of $y_i$ for which $w(uy_i)\geq 2s-s/k$ must also be in $X$. Hence, $y_i\in X$ implies $y_{i-1},y_{i+1},z_i\in X$ (or $y_n,y_{2},z_2\in X$ in case $i=1$ or $y_{n-1},y_1,z_n\in X$ in case $i=n$). Combining this with our earlier argument that $X\cap \{y_1,\dots,y_k\}$ cannot be empty we obtain that $\{y_1,\dots,y_k,z_1,\dots,z_k\}$ must be a subset of $X$.
		
		We now define $Q_j:=\{a_i: v_i^j\in X\}$ for every $j\in [k]$. First observe that $Q_j\cap Q_{j'}=\emptyset$ for every pair $j,j'\in [k]$ as the edge $v_i^jv_i^{j'}$ having  weight $\rho<-(a_i+\epsilon)$ prevents the simultaneous containment of $v_i^j$ and $v_i^{j'}$ in $X$. We now argue that $|Q_i|=s/k$. For this observe that the utility of $y_j$ in the coalition $\{y_1,\dots,y_k,z_1,\dots,z_k\}$ is $6s-s/k+\delta$. Hence $\sum_{i\in [n],\atop{v_i^j\in X}}a_i+|\{i:v_i^j\in X\}|\epsilon> s/k-\frac{1}{4}$. Since $|\{i:v_i^j\in X\}|\epsilon\leq \frac{1}{2}$ and the $a_i$'s are integer we get that $\sum_{a_i\in Q_j}a_i=\sum_{i\in n,\atop{v_i^j\in X}}a_i\geq s/k$. On the other hand, the utility of $z_j$ in the coalition $\{y_1,\dots,y_k,z_1,\dots,z_k\}$ is $s/k+\delta$.   Since all edges of the form $z_j v_i^j$ have weight $-a_i$ this implies that $\sum_{i\in [n],\atop{v_i^j\in X}}a_i< s/k+\frac{1}{4}$. Since all $a_i$'s are integer this condition is equivalent to $\sum_{a_i\in Q_j}a_i=\sum_{i\in [n],\atop{v_i^j\in X}}a_i\geq s/k$. Hence we have argued that $\mathcal{Q}$ is a partition of $A$ into $k$ parts such that $\sum_{a_i\in Q_j}a_i=s/k$ for every $j\in k$.\\
		
		Finally, observe that the graph $G$ has vertex integrity $3k$ and cluster deletion number $2k$ which is witnessed by deletion set $\{y_1,\dots,y_k,z_1,\dots,z_k\}$.
	\end{proof}
\end{toappendix}

We prove the third hardness result by a reduction from \textsc{Bounded Degree Deletion}.

\begin{theorem}\label{thm:csv:td}\ifthenelse{\boolean{short}}{\textup{($\star$)}}{}
	\textsc{Core Stability Verification} is coW[1]-hard parameterized by
	tree-depth, even if all weights are in $\{-1,1\}$.
	
\end{theorem}

\begin{toappendix}
	\begin{proof}[Proof \ifthenelse{\boolean{short}}{of \cref{thm:csv:td}}{}]
		We present a reduction from \textsc{Bounded Degree Deletion}, which is known to
		be W[1]-hard parameterized by tree-depth \cite{GanianKO21,LampisV23}. In this
		problem we are given a graph $G$ and an integer $\Delta^*$ and are asked to
		find the largest induced subgraph of $G$ that has maximum degree at most
		$\Delta^*$ (equivalently, we seek the smallest set of vertices whose deletion
		makes the maximum degree of $G$ at most $\Delta^*$).
		
		Let $(G,\Delta^*,s)$ be an instance of \textsc{Bounded Degree Deletion} and
		suppose we want to decide if $G$ contains an induced subgraph of maximum degree
		at most $\Delta^*$ with at least $s$ vertices.  We construct an instance of
		\textsc{CSV} as follows: we keep $G$ and we assign every edge of $G$ weight
		$-1$; we construct two new vertices $x,x'$; we construct $s(\Delta^*+1)-1$ new
		vertices and connect each to both $x$ and $x'$ with edges of weight $1$; for
		each $u\in V(G)$ we construct $\Delta^*+1$ new vertices
		$u_1,\ldots,u_{\Delta^*+1}$ and connect them to $x$ and $u$ with edges of
		weight $1$; finally we attach to each $u_i$, for $u\in V(G), i\in[\Delta^*+1]$
		a leaf $u_i'$ connected to it via an edge of weight $1$. The initial partition
		we wish to check for stability places $x, x'$ and their common neighbors
		together; for each $u\in V(G), i\in[\Delta^*+1]$ $u_i, u_i'$ are together;
		while all other vertices are singletons.
		
		It is not hard to see that the new graph has essentially the same tree-depth as
		$G$, as deleting $x,x'$ gives us $G$ with some paths of length $2$ attached to
		each vertex and some isolated vertices.
		
		For the forward direction, suppose that $G$ does contain a set of vertices $S$
		such that $|S|\ge s$ and $G[S]$ has maximum degree at most $\Delta^*$. We claim
		that if we take $S\cup\{x\}$ and add to it all $u_i$ such that $u\in S$, then
		we obtain a blocking coalition. Indeed, the utility of $x$ in the initial
		partition is $s(\Delta^*+1)-1$, while in the new coalition it is
		$|S|(\Delta^*+1)\ge s(\Delta^*+1)$. Furthermore, all vertices $u_i$ of the new
		coalition previously had utility $1$, whereas in the new coalition they have
		utility $2$. Finally, each $u\in S$ now has utility $\Delta^*+1-|N(u)\cap S|\ge
		1$ while previously $u$ had utility $0$.
		
		For the converse direction, suppose that a blocking coalition $X$ exists. Then,
		$x'$ cannot belong in such a coalition, as its utility is already maximum, and
		neither can the common neighbors of $x$ and $x'$. Similarly, $u_i'$ cannot be
		in $X$. If $x$ does not belong in $X$, then no $u_i$ can be in $X$, as the
		maximum utility she could obtain is $1$, which is the same as her initial
		utility.  As all remaining vertices only have negative utilities to each other,
		this would lead to a contradiction. Therefore, $x\in X$.  However, for $x$ to
		increase her utility, it must be the case that $X$ contains at least $s$
		vertices of $G$. To see this, observe that if $u\not\in X$, then $u_i\not\in X$
		for all $i\in[\Delta^*]$. However, we need the utility of $x$ in $X$ to be at
		least $s(\Delta^*+1)$, and since from each $u\in S$ we have at most
		$\Delta^*+1$ vertices $u_i$ which are in $X$, we conclude that $|X\cap V(G)|\ge
		s$. We now argue that $G[X\cap V(G)]$ has maximum degree at most $\Delta^*$.
		Indeed, if $u\in X\cap V(G)$ has at least $\Delta^*+1$ neighbors in $X\cap
		V(G)$, since each such neighbor contributes $-1$ to the utility of $u$, $u$
		cannot have strictly positive utility. We therefore have that $X\cap V(G)$ is a
		valid solution to the original instance.  \end{proof}
\end{toappendix}

\subsection{Algorithms}
In this section, we prove the algorithmic results complementing the hardness results of the previous section. %The first result is proved by a standard dynamic programming approach.
\begin{theorem}\label{thm:csv:trees} \ifthenelse{\boolean{short}}{\textup{($\star$)}}{}
	\textsc{Core Stability Verification} is polynomial time solvable on trees.
\end{theorem}
\begin{toappendix}
	\begin{proof}[Proof\ifthenelse{\boolean{short}}{ of \cref{thm:csv:trees}}{}] 
		This can be shown using a simple bottom-up dynamic programming approach. Let $(T,w)$ be a weighted tree and $\mathcal{P}$ a coalition structure of $(T,w)$. We arbitrarily pick a root vertex $r\in V(T)$ to define an ancestor/descendant relationship on $T$.
		For any $u\in V(T)$ we let $T_u$ be the sub-tree of $T$ rooted at $u$. We say that a set $X\subseteq V(T_u)$ with $u\in X$ is a downwards blocking coalition  if $\ut(X,v)> \ut_\mathcal{P}(v)$ for every vertex $v\in X$, $v\not=u$.
		For any $u\in V(T)$ we define $\ut_{\below}(u)$ to be the maximum utility $u$ can achieve in any downwards blocking coalition $X\subseteq V(T_u)$ which contains $u$.
		
		First note that $\ut_{\below}(u)$ of any leaf of $T$ is $0$. Furthermore, if $u$ is a vertex, $X\subseteq V(T_u)$ a downwards blocking coalition with $u\in X$ such that $\ut(X,u)=\ut_{\below}(u)$  and  $v$ a child of $u$ with $w(uv)<0$ then $v\notin X$. That is because otherwise $X\setminus V(T_v)$ will yield a downwards blocking coalition with larger utility for $u$.
		Let $u$ be a non-leaf vertex and $v_1,\dots ,v_\ell$ its children.
		Note that $v_i$ can be contained in a downward blocking coalition $X\subseteq V(T_u)$ with $u\in X$ if and only if $\ut_{\below}(v_i)+w(uv_i)>\ut_\mathcal{P}(v_i)$. Hence,
		$$\ut_{\below}(u)=\sum_{i\in [\ell]\atop{ w(uv_i)\geq 0\atop{\ut_{\below}(v_i)+w(uv_i)>\ut_\mathcal{P}(v_i)}}} w(uv_i).$$
		Now note that a downwards blocking coalition $X\subseteq V(T_u)$ with $u\in X$ is a blocking coalition of $\mathcal{P}$ if and only if $\ut_{\below}(u)>\ut_\mathcal{P}(u)$. Furthermore, there is a blocking coalition $X$ of $\mathcal{P}$ if and only if there is $u\in V(T)$ such that $X\subseteq V(T_u)$ with $u\in X$ is a downwards blocking coalition and $\ut_{\below}(u)>\ut_\mathcal{P}(u)$. That is as we can assume that any blocking coalition $X$ is connected and hence $X$ must contain a vertex $u$ which is an ancestor of every vertex $v\in X$.
		
		Hence, a bottom-up dynamic programming algorithm computing $\ut_{\below}(u)$ for every $u\in V(T)$ and then verifying whether there is $u\in V(T)$ with $\ut_{\below}(u)>\ut_\mathcal{P}(u)$ yields a linear time algorithm for \textsc{CSV}.
	\end{proof}
\end{toappendix}

\begin{theorem}\label{thm:csv:FPT:vc}
	\textsc{Core Stability Verification} can be solved in time $(\vc w_{\max})^{O(\vc)}\Delta^2+O(\vc n)$.
\end{theorem}

\begin{proof}
	Given a graph $(G,w)$ and a partition $\PP$ of $V(G)$, we check whether there is a blocking coalition $X\subseteq V(G)$ of $\PP$.
	In the algorithm, we first compute a minimum vertex cover $S$ of size $\vc$ in time $O(1.2738^{\vc}+\vc n)$~\cite{ChenKX10}.
	Then we guess an intersection of $X$ and $S$. The number of possible candidates of the intersection is at most $2^{|S|}$.
	Let $I'\subseteq V(G)\setminus S$ be the set of vertices in $V(G)\setminus S$ such that each vertex $u\in I'$  satisfies $\sum_{v\in N_G(u)\cap X\cap S} w(uv)>{\ut}_{\PP}(u)$. These vertices in $I'$ could form a blocking coalition of $\PP$ by cooperating with the vertices in $X\cap S$.
	In order for $X$ to become a blocking coalition, all the vertices in $X\cap S$ must have larger utility in $X$ than their utility in $\PP$ after some vertices in $I'$ joined $X$. This condition can be represented in Integer Linear Programming (ILP) as follows:
	\begin{align}\label{eq:ILP:core}
		\sum_{v\in I'} w(uv)x_v +\sum_{v\in N_G(u)\cap X\cap S}w(uv)  \ge \ut_{\PP}(u)+ 1 \ \ \ & \forall u\in X\cap S\\
		x_v\in \{0,1\} \ \ \ & \forall v\in I' 
	\end{align}
	where the variable $x_v$ represents whether vertex $v\in I'$ joins $X$. In the left hand side of (\ref{eq:ILP:core}), $\sum_{v\in N_G(u)\cap X\cap S}w(uv)$ represents the contribution of edge weights in $X\cap S$ to the utility of $u$. Clearly, the ILP is feasible if and only if there is a blocking coalition $X$ because the utility of each agent in $X$ strictly increases. Note that we suppose that each edge weight is an integer.
	
	Here, the feasibility check of an ILP $\{Ax=b, x\ge 0, x\in \mathbb{Z}^n\}$ can be solved in time $(m\cdot \Lambda)^{O(m)}\cdot ||b||^{2}_{\infty}$ where $A\in \mathbb{Z}^{n\times m}$, $b\in \mathbb{Z}^m$, and $\Lambda$ is an upper bound on each absolute value of an entry in $A$~\cite{EisenbrandW20}.
	Since the maximum absolute value of coefficients of variables is $w_{\max}$, the range of $\ut_{\PP}(u)$ and $\sum_{v\in N_G(u)\cap X\cap S}w(uv)$  is $[-w_{\max}\Delta,w_{\max}\Delta]$, and the number of constraints is at most $|S|$, the ILP can be solved in time $(\vc w_{\max})^{O(\vc)}(w_{\max}\Delta)^2 = (\vc w_{\max})^{O(\vc)}\Delta^2$.
	Thus, the total running time is $O(1.2738^{\vc}+\vc n)+2^{\vc}(\vc w_{\max})^{O(\vc)}\Delta^2=(\vc w_{\max})^{O(\vc)}\Delta^2+O(\vc n)$.
\end{proof}

\begin{theorem}\label{thm:csv:FPT:vi}\ifthenelse{\boolean{short}}{\textup{($\star$)}}{}
	\textsc{Core Stability Verification} is FPT parameterized by
	$\vi+w_{\max}$.
	
\end{theorem}
%\todo[inline]{T.H wrote the proof. I think we consider not partitions of a component, but subsets. Is that correct? N: No we need partitions and then also subsets. I will point it out at the appropriate place why subsets is not enough.}
\begin{toappendix}
	\begin{proof}[Proof\ifthenelse{\boolean{short}}{ of \cref{thm:csv:FPT:vi}}{}]
		As with Theorem~\ref{thm:csv:FPT:vc}, given a graph $(G,w)$ and a partition $\PP$, we check whether there is a blocking coalition $X\subseteq V(G)$ of $\PP$.
		Let $S$ be a $\vi(k)$-set $S$ where $k=\vi$. One can find a $\vi(k)$-set $S$ in time $O(k^{k+1}n)$~\cite{DrangeDH16}.
		We define \emph{types} of components of $G[V(G)\setminus S]$ with respect to isomorphisms, weights and the partition $\mathcal{P}$. %More precisely, we say that two components $C_1$ and $C_2$ have the same type, if $G[C_1\cup S]$ and $G[C_2\cup S]$ are isomorphic and the corresponding edges have the same weight. 
		%Note that $S$ is common in $G[C_1\cup S]$ and $G[C_2\cup S]$, which means that $G[E(C_1\cup S)\setminus E(S)]$ and $G[E(C_1\cup S)\setminus E(S)]$ are isomorphic with respect to weights, where $G[E']$ is a graph induced by an edge subset $E'$ and $E(S)$ is a set of edges whose endpoints are in $S$.\todo[inline]{N: I suggest the following definition of types: 
			More precisely, we say that two components $C_1$ and $C_2$ of $G[V(G)\setminus S]$ have the same type, if there is an isomorphism $\iota$ from $G[C_1\cup S]$ to $G[C_2\cup S]$ such that $\iota(s)=s$ for every $s\in S$, corresponding edges have the same weight and for every $P\in \mathcal{P}$ there is $Q\in \mathcal{P}$ such that $\iota(P\cap C_1)=Q\cap C_2$.
			Let $\mathcal{T}$ be a set of representatives of all component types appearing in $G[V(G)\setminus S]$ and $c_t$ the number of components of $G[V(G)\setminus S]$ of type $t$ for every $t\in \mathcal{T}$.
			In addition, we define $\mathcal{X}_t$ to be the set of subsets of $V(t)$.
			
			We are ready to design an algorithm for \textsc{CSV}. 
			First, we guess an intersection of a blocking coalition $X$ and the $\vi(k)$-set $S$. The number of subsets of $S$ is at most $2^{k}$. In the next step, we formulate \textsc{CSV} as an ILP for each intersection $X\cap S$. First, we define variables $x_{(t,X_t)}$ for any type $t$ and any $X_t\in \mathcal{X}_t$. The variable $x_{(t,X_t)}$ represents the number of components of type $t$ in $G[V(G)\setminus S]$ such that the blocking coalition $X$ intersects the component in the set $Q_t$ corresponding to $X_t$ under the isomophism witnessing the type.
			Since the number of non-isomorphic components is at most  $(w_{\max})^{O(k^2)}$, the number of partitions of the vertices of a component is at most $2^{k^2}$ (any partition corresponds to the subsets of edges incident to vertices in the same part of the partition)  and the size of $\mathcal{X}_t$ for a type $t$ is at most $2^k$, the number of variables is at most $(w_{\max})^{O(k^2)}$. 
			For a vertex $u\in X\cap S$, a type $t\in \mathcal{T}$, and a set $X_t\in \mathcal{X}_t$, we denote by $w_{(t,X_t,u)}$ the sum of weights of edges incident to $u$ and a vertex in $X_t$. It represents the increase of the utility of $u$ if $X_t$ joins the blocking coalition $X$.
			
			Let $\mathcal{X}'_t\subseteq \mathcal{X}_t$ be the set of subsets of $V(t)$  such that for every $X_t\in \mathcal{X}'_t$, each vertex in $X_t$ has higher utility in $X_t\cup (X\cap S)$ than in the initial coalition $\PP$. %\todo{N: As your type does not take into account $\mathcal{P}$, a set $X_t$ can be good for some components of type $t$ but bad for others. }. 
			Moreover, let $\mathcal{T}'\subseteq \mathcal{T}$ be the set of types $t$ satisfying that  $\mathcal{X}'_t\neq \emptyset$. %These guarantee that $X$ is a blocking coalition. 
			Then deciding whether there exists a blocking coalition can be formulated as the following ILP.
			
			\begin{align*}
				\sum_{X_t\in \mathcal{X}'_t} x_{(t,X_t)} &= c_t   & \forall t\in \mathcal{T}', \\
				\sum_{t\in \mathcal{T}'} \sum_{X_t\in \mathcal{X}'_t} w_{(t,X_t,u)} x_{(t,X_t)} +\sum_{v\in N_G(u)\cap X\cap S}w(uv)  &> \ut_{\PP}(u)    & \forall u\in X\cap S. \\
			\end{align*}
			
			The first equality encodes the constraint that the numbers $X_(t,X_t)$, $X_t\in \mathcal{X}_t$ should correspond to the sizes of the parts of a partition of the $c_t$ components of type $t$. The second inequality represents the constraint that each vertex $u\in X\cap S$ has larger utility than the utility of $u$ in $\PP$.
			It is clear that the ILP is feasible if and only if there exists a blocking coalition.
			
			For the running time of the algorithm, we first compute a $\vi(k)$-set $S$ in time $O(k^{k+1}n)$. For $S$, we guess $2^k$ intersections $X\cap S$, determine the type count $c_t$ for every type $t$ in time $O(k^{k}n)$ and solve the ILP for each $X\cap S$. Given an intersection $X\cap S$, we can compute $w_{(t,X_t,u)}$  and $\sum_{v\in N_G(u)\cap X\cap S}w(vu)$ can be computed in polynomial time. 
			Finally, it is known that the feasibility of ILP with $p$ variables can be computed in FPT time parameterized by $p$~\cite{Lenstra83,FrankT87,Kannan87}.
			Since the number of variables is at most $(w_{\max})^{O(k^2)}$, one can solve \textsc{CSV} in time $f(k+w_{\max})n^{O(1)}$ where $f$ is some computable function.
		\end{proof}
	\end{toappendix}
	
	\begin{theorem}\label{thm:tw:XP}\ifthenelse{\boolean{short}}{\textup{($\star$)}}{}
		\textsc{Core Stability Verification} can be computed in time $(\Delta w_{\max})^{O(\tw)}n^{O(1)}$.
	\end{theorem}
	
	\begin{toappendix}
		\begin{proof}[Proof\ifthenelse{\boolean{short}}{ of \cref{thm:tw:XP}}{}]
			Our algorithm relies on standard techniques (dynamic programming
			over tree decompositions) so we sketch some of the details. We are given as
			input an edge-weighted graph $(G,w)$, a nice tree decomposition of $G$, and an
			initial partition $\PP$ of $V(G)$. We are asked if there exists a blocking
			coalition $X\subseteq V(G)$ such that all vertices $u\in X$ have a strictly
			higher utility in $X$ than in their coalition in $\PP$. We begin by calculating
			for each $u\in V(G)$ the utility $u$ gains in the initial partition $\PP$. This
			can clearly be done in polynomial time in the size of the input, that is,
			polynomial in $n+\log w_{\max}$.
			
			Suppose now that the given nice tree decomposition is rooted at some arbitrary bag and consider a bag $B$. In order to describe the dynamic programming table that our algorithm needs to maintain for $B$, we define the signature of a possible blocking coalition $X$ as a pair consisting of (i) the set $X\cap B$ (ii) a function $\ut : X\cap B \to \mathbb{Z}$ which describes for each vertex $u$ that belongs in the bag $B$ and the blocking coalition $X$ how much utility $u$ gains from vertices that appear in bags below $B$ in the tree decomposition (but not in $B$). The dynamic programming algorithm maintains an entry for each possible signature that tells us if there exists a coalition $X$ that matches this signature and is feasible below $B$, that is, satisfying the criterion that all vertices which appear only below $B$ in the decomposition and belong in $X$ have strictly higher utility in $X$ than in their coalition of $\PP$.
			
			Given the above it is now straightforward to implement the DP algorithm: for Introduce nodes we consider two cases to produce the signatures of the new bag (the introduced vertex belongs in $X$ or not); for Join nodes we take all pairs of signatures that agree on $X\cap B$ and simply add their utility functions; while for Forget nodes, if we are forgetting vertex $u$ we need to discard signatures where $u\in X$ but the utility of $X$ is not strictly higher than in $\PP$, while for other signatures we adjust the utilities of the neighbors of $u$ accordingly.
			
			What remains is to estimate the running time of the algorithm. This is dominated by the number of possible signatures, which is at most $2^{\tw}(2\Delta w_{\max}+1)^{\tw}$. Here, the first factor is for storing a subset of $B$, and the second is for storing the utility of each vertex, which must necessarily be in the range $[-\Delta w_{\max}, \Delta w_{\max}]$.
		\end{proof}
	\end{toappendix}
	
	\begin{theorem}\label{thm:tw+d:FPT} \ifthenelse{\boolean{short}}{\textup{($\star$)}}{} \textsc{Core Stability Verification} can be
		computed in time $2^{O(\tw\Delta)}(n+\log w_{\max})^{O(1)}$.  \end{theorem}
	\begin{toappendix}
		\begin{proof}[Proof\ifthenelse{\boolean{short}}{ of \cref{thm:tw+d:FPT}}{}]
			
			This is obtained using the same algorithm as in Theorem \ref{thm:tw:XP} using the following observation: in the analysis of the algorithm of Theorem \ref{thm:tw:XP} we have a DP table which stores for each vertex $u$ of $X\cap B$ the utility that $u$ gains in the blocking coalition $X$. The range of possible values was then bounded by $[-\Delta w_{\max}, \Delta w_{\max}]$. However, we can also observe that the set of possible utilities of a vertex $u$ with degree $\Delta$ has size at most $2^{\Delta}$, since the utility of $u$ is obtained by the subset of its neighbors which are contained in $X$. As a result, our algorithm may consider only the $2^{\Delta}$ possible utility values for each vertex of the bag, giving $2^{\Delta\tw}$ possible distinct solution signatures. The rest of the analysis is identical to that of Theorem \ref{thm:tw:XP}.
		\end{proof}
	\end{toappendix}

\section{Core Stability}

In this section we study the complexity of  \textsc{Core Stability} (\textsc{CS}). We first show that \textsc{CS} remains $\Sigma_2^p$-complete even on graphs of bounded vertex cover number (\cref{thm:csfSigma2}). 

The second part of this section is dedicated to extending our understanding of
the complexity of \textsc{CS} parameterized by $\tw+\Delta$. We
give an algorithm for \textsc{CS} running in time
$2^{2^{\mathcal{O}(\Delta\tw^2)}}n$ (\cref{thm:doubleExpAlgoCS}) and
improving on the previous algorithm based on Courcelle's Theorem by Peters
\cite{Peters16a}. In order to avoid having to formulate a tedious dynamic
programming algorithm, we instead obtain our algorithm via a reduction to
$\exists\forall$-SAT, which is known to be solvable in double-exponential time
(in treewidth) \cite{Chen04}.  We complement these results by giving an ETH
based lower bound of $2^{2^{o(\pw)}}$ on graphs of bounded degree
(\cref{thm:ETHlowerBoundCS}). This shows that the double-exponential dependence
of our algorithm on treewidth is in fact inevitable, and confirms a pattern
shown by other $\Sigma_2^p$-complete problems \cite{LampisM17}.

Before we proceed, we describe a gadget that  appears in
several reductions.

\subsection{An Auxiliary Gadget}\label{sec:auxiliaryGadget} 

For reductions to \textsc{CS} we use a gadget to control what core
stable partitions (if any exist) have to look like. For this we rely on an
example given in \cite[Example 1]{AzizBS13} which shows a small concrete graph
$H$ (on six vertices) that does not admit any core stable partition and is
minimal for this property.  The intuitive idea is that we \emph{attach}
copies of $H$ onto various vertices of our construction with the aim of forcing
the vertex onto which $H$ is attached to be in the same partition as some
vertex of $H$.  Because by assumption it is impossible to partition $H$ in a
stable way, the vertex on which a copy of $H$ is attached \emph{must} be placed
in the same coalition as some vertex of $H$. Building on this idea, we 
also attach a copy of $H$ on a set of vertices of $G$, forcing at least one of
them in the same partition as a vertex of $H$.

To make this idea formally precise, we will use the following minimally non-stable graph
$(H,w_\rho)$ where $\rho$ is some sufficiently small integer.  Let $H$ be $K_6$
with vertex set $\{h,h_1,\dots,h_5\}$. We set
$w_\rho(hh_1)=w_\rho(h_2h_3)=w_\rho(h_4h_5)=5$, $w_\rho(h_1h_2)=
w_\rho(h_3h_4)=w_\rho(h_5h)=4$ and
$w_\rho(h_1h_3)=w_\rho(h_3h_5)=w_\rho(h_5h_1)=3$. We let all remaining edges
have weight $\rho$.  Using the same argument as \cite{AzizBS13} we can show
that $(H,w_\rho)$ is not core stable if $\rho<-15$.  We further observe that
for the graph $(H\setminus\{h\},w_\rho|_{\{h_1,\dots,h_5\}})$ the partition
$(\{h_1,h_2,h_3\},\{h_4,h_5\})$ is core stable if $\rho<-15$. We denote this
partition of $(H\setminus\{h\},w_\rho|_{\{h_1,\dots,h_5\}})$ by
$\mathcal{P}_{(H,w_\rho)}$.

For a graph $(G,w)$, $S\subseteq V(G)$ an independent set in $G$ and $\rho,\xi$ integers, we let $(G',w')$ be the graph obtained in the following way. The graph $(G',w')$ consist of the disjoint union of $(G,w)$ and $(H,w_\rho)$ to which we add the edges $hs$ of weight $\xi$ for every $s\in S$ and  edges $h_is$ and $ss'$  of weight $\rho$ for every $i\in [5]$, $s\not=s'\in S$.
We say that $(G',w')$ is obtained from $G$ by \emph{$\xi$-attaching  $(H,w_\rho)$ at $S$}.  In the next lemma we show that attaching the gadget $(H,w_\rho)$ enforces certain properties for coalitions.

\begin{lemmarep}\label{lem:auxiliaryGadget}
	Let $(G,w)$ be a graph, $S\subseteq V(G)$ an independent set in $G$, $\rho$ an integer smaller than $- \max_{u\in V(G')}\sum_{uv\in E(G'),\atop{w'(uv)>0}}w(uv)$, $\xi\geq 9$ and $(G',w')$ the graph obtained from $(G,w)$ by $\xi$-attaching 
	$(H,w_\rho)$ at $S$. The graph $(G',w')$ has the following properties.
	\begin{enumerate}[left=5pt , label=$(\Pi\arabic*)$]
		\item\label{prop:gadgetP1}  For every connected core stable partition $\mathcal{P}$ of $(G',w')$  there is $s\in S$ and $P\in \mathcal{P}$ such that $\{s,h\}\in P$ and $s'\notin P$ for every $s'\in S$, $s'\not=s$.
		\item\label{prop:gadgetP2} For every connected core stable partition $\mathcal{P}$ of $(G',w')$ and every part $P\in \mathcal{P}$ with $P\cap \{h_1,\dots,h_5\}\not=\emptyset$ we have that $P\subseteq \{h_1,\dots,h_5\}$.
		\item\label{prop:gadgetP3} For every partition $\mathcal{P}$ of $(G',w')$ with $\mathcal{P}_{(H,w_\rho)}\subseteq \mathcal{P}$ and $\{h,s\}\subseteq P$ for some $s\in S$ and some $P\in \mathcal{P}$,
		every blocking coalition $X$ is disjoint from $V(H)$.
	\end{enumerate}
	Further, the properties are satisfied even if we add additional edges of weight $\rho$ to $(G',w')$.
\end{lemmarep}
\begin{proof}
	To argue the furthermore part, we assume that we arbitrarily added edges of weight $\rho$ to $(G',w')$. Note that we cannot add any edges between vertices of $V(H)$ as $H$ is complete.
	Let $\mathcal{P}$ be a connected core stable partition of $(G',w')$.
	First observe that no two vertices $u,v\in V(G')$ with $w(uv)=\rho$ can be contained in the same part of the partition $\mathcal{P}$ due to the assumption that $\rho<- \max_{u\in V(G')}\sum_{uv\in E(G'),\atop{w'(uv)>0}}w(uv)$ and hence the utility of $u$ and $v$ in $\mathcal{P}$ would be negative. Similarly, no two vertices $u,v\in V(G')$ with $w(uv)=\rho$ can both be contained in a blocking coalition.
	
	Let $P\in \mathcal{P}$ be the part containing $h$.
	We first argue that  any part $P'\in \mathcal{P}$, $P'\not=P$ (and hence $h\notin P'$) containing some $h_i$, $i\in [5]$ is a subset of $\{h_1,\dots,h_5\}$. If $P'$ is such a part, then $P'$ cannot contain any vertex from  $N^\rho(H)=\{v\in V(G):w(\tilde{h}v)=\rho \text{ for some } \tilde{h}\in V(H)\}$. Since removing $h$ and $N^\rho$ separates the vertices   $\{h_1,\dots,h_5\}$ from the rest of the graph and we assumed that $G[P']$ is connected, we conclude that $P'\subseteq \{h_1,\dots,h_5\}$.
	Since $(H,w_\rho)$ is not core stable this implies that the part $P$ must  contain some vertex $v\in V(G)$.  But since we assumed that $G[P]$ is connected  and the removal of $S$ and $\{v\in V(G):w(vh)=\rho\}$ separates the set of vertices $V(H)$ from the rest of the graph, $P$ must contain some vertex from $S$. Since $w(ss')=\rho$ for any $s,s'\in S$, $s\not=s'$, this implies that  property~\ref{prop:gadgetP1}  is true. Additionally, observe that $P\cap S\not=\emptyset$ imply that $P\cap \{h_1,\dots,h_5\}=\emptyset$  and hence property~\ref{prop:gadgetP2} follows from our observation that  $P'\subseteq \{h_1,\dots,h_5\}$ for every part $P'\in \mathcal{P}$, $P'\not=P$ containing some $h_i$, $i\in [5]$.
	
	To prove property~\ref{prop:gadgetP3} assume that $\mathcal{P}$ is a partition of $V(G')$ with $\{h_1,h_2,h_3\}, \{h_4,h_5\}\in \mathcal{P}$ and let $P\in \mathcal{P}$ be  a part and $s\in S$ be an element such that $\{h,s\}\subseteq P$. Assume $X$ is a blocking coalition.  Towards a contradiction, assume that $h\in X$. First observe that $h$ has utility $\xi$ in $\mathcal{P}$. Furthermore, $h$ is incident to two edges of positive weight in $H$ and these two weights sum up to $9\leq \xi$. 
	Hence, there has to be some vertex $u\in V(G)$ which is contained in $X$ and $w(hu)>0$. By construction, this implies that $u$ has to be a vertex in $S$. But then no other vertex $v\in S\cup \{h_1,h_5\}$ can be contained in $X$ as  $w(uv)=\rho$ and hence the utility of $h$ in $X$ can be at most $\xi$. Therefore, $h\notin X$. Now observe that $h_2$ cannot join $X$ because its utility in $\mathcal{P}$ is the sum of its positive weight edges. But then
	$h_1$ cannot join because the sum of its positive edges excluding edges incident to $h$ and $h_2$ is $6$ and $h_2$ has utility $7$ in $\mathcal{P}$. But then $h_3$ cannot join $X$ as its the sum of its positive edges excluding edges incident to $h_1$ and $h_2$ is $7$ while $h_3$ has utility $8$ in $\mathcal{P}$. Finally, now $h_4$ and $h_5$ cannot join $X$ as $\{h_4,h_5\}\subseteq X$ and hence they cannot improve their utility without further vertices from $H$. This concludes the proof of property~\ref{prop:gadgetP3}.  
\end{proof}
We can also use the gadget in a slightly different way. Details  are deferred to the Appendix.
\begin{toappendix}
	For a graph $(G,w)$, $S\subseteq V(G)$ an independent set in $G$, $\rho< -15$ and $\xi\geq 9$ we say that $(G',w')$ is obtained from $(G,w)$ by \emph{$\xi$-neighborhood attaching $(H,w_\rho)$ at $S$} if $(G',w')$ is obtained from $(G,w)$ by $\xi$-attaching $(H,w_\rho)$ at $S$ and adding edges $hv$ of weight $\rho$ for every $v\in N_G(S)\setminus S$.
	\begin{lemma}\label{lem:auxiliaryNeighborhoodGadget}
		Let $(G,w)$ be a graph, $S\subseteq V(G)$ an independent set, $\rho$ and integer smaller than $- \max_{u\in V(G')}\sum_{uv\in E(G'),\atop{w'(uv)>0}}w(uv)$, $\xi\geq 9$ and $(G',w')$ the graph obtained from $(G,w)$ by $\xi$-neighborhood attaching 
		$(H,w_\rho)$ at $S$. The graph $(G',w')$ has properties~\ref{prop:gadgetP2}, \ref{prop:gadgetP3} and the following property.
		\begin{enumerate}[left=5pt , label=$(\Pi'\arabic*)$]
			\item\label{prop:gadgetP1Alternative}  For every connected core stable partition $\mathcal{P}$ of $(G',w')$  there is $s\in S$ such that $\{h,s\}\in \mathcal{P}$.
		\end{enumerate}
	\end{lemma}
	\begin{proof}
		Since $(G',w')$ was obtained from $(G,w)$ be $\xi$-attaching $(H,w_\rho)$ at $S$ and adding some edges of weight $\rho$ we get that \ref{prop:gadgetP1},\ref{prop:gadgetP2} and \ref{prop:gadgetP3} are satisfied. Assume $\mathcal{P}$ is a connected core stable partition and $s\in S$ is an element and $P\in \mathcal{P}$ a part such that $\{s,h\}\subseteq P$. Towards a contradiction assume that there is $u\in P$ such that $u\not= s$, $u\not= h$. By property \ref{prop:gadgetP1} and \ref{prop:gadgetP2} we know that $u\in V(G)\setminus S$. As $h$ has no neighbors in  $V(G)\setminus S$ and $G[P]$ is connected, this implies that $P$ must contain some vertex $v\in V(G)\setminus S$ which is adjacent to $s$. But this cannot happen as $w(hv)=\rho$. Therefore, $P=\{h,s\}$ proving property \ref{prop:gadgetP1Alternative}. 
	\end{proof}
	
	Observe that if we attach multiple copies of $(H,w_\rho)$ to some graph $G$, then the resulting graph does not depend on the order in which we attached copies of $(H,w_\rho)$ as long as the sets, at which we attached the copies, are disjoint.
\end{toappendix}
\subsection{Core stability on graphs of bounded vertex cover number}
In this section we prove the following result.
\begin{theorem}\label{thm:csfSigma2}
	\textsc{Core Stability}  is $\Sigma_2^p$-complete on graphs of bounded vertex cover number.
\end{theorem}
To obtain this result we use a variation of the constructions used to prove \cref{thm:csv:vc:weak} and \cref{thm:csv:vi} however reducing from an appropriate variant of \textsc{Partition}.
%\todo[inline]{ML: Is this really correct regarding the polynomial bound on the weights? I would think we need exponentially large weights for this result...N: Yes, you are right.}

\begin{proof}
	\begin{figure}
		\centering
		\includegraphics[scale=0.8]{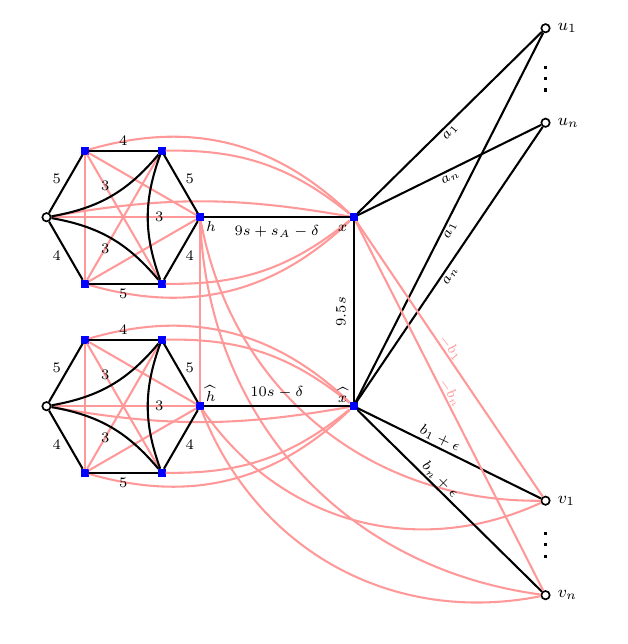}
		\caption{The graph $(G,w)$ constructed in the proof of Theorem~\ref{thm:csfSigma2}. Here all edges without label have weight $\rho$. Vertex cover vertices are depicted by blue squares.}
		\label{fig:EAPartition-cs}
	\end{figure}
	First observe that \textsc{CS} is in $\Sigma_2^p$ since for  any guessed coalition structure $\mathcal{P}$ of $(G,w)$ deciding whether $\mathcal{P}$ is core stable is in coNP. That is, given a  blocking coalition $X$ we can verify in polynomial time (by comparing utilities in $X$ and $\mathcal{P}$) whether $\mathcal{P}$ is not core stable.
	
	We show $\Sigma_2^p$-hardness by a reduction from \textsc{$\exists\forall$-partition}. Given two sets of positive integers $A=\{a_1, \ldots, a_n\}$ and $B=\{b_1, \ldots, b_n\}$, the \textsc{$\exists\forall$-partition} problem asks whether there exists a subset $A'\subseteq A$ such that for every subset $B'\subseteq B$ it holds that $\sum_{c\in A'\cup B'}c \not=s/2$ where $s=\sum_{c\in A\cup B} c$.  We let $s_A=\sum_{a\in A} a$ and $s_B=\sum_{b\in B} b$.
	Note that a variant of \textsc{$\exists\forall$-partition} was shown to be $\Sigma_2^p$-hard \cite[Lemma 6.2.]{berman1997complexity} where a target value $t$ is given additionally and the question is whether there exists $A'\subseteq A$ such that for every $B'\subseteq B$ it holds that $\sum_{c\in A'\cup B'}c \not=t$. We can reduce this variant to \textsc{$\exists\forall$-partition} by adding integer $2s-t$ to $A$ and $s+t$ to $B$ where $s=\sum_{c\in A\cup B} c$. 
	
	We construct an instance $(G,w)$ of \textsc{CS} as follows. First we let $\delta>0$ be an integer smaller than $\min_{i\in [n]}\{a_i,b_i\}$ and $\epsilon>0$ be an integer such that $\epsilon<\delta/n$. We can assume that such integers exist by picking a suitable positive integer and multiplying all elements in $A\cup B$ by it. We further set $\rho=-30s$. 
	We  introduce vertices $x,\widehat{x}$ and add edge $x\widehat{x}$ of weight $9.5s$. 
	We further introduce vertices $u_i,v_i$ for every $i\in [n]$ and add edges $u_ix$, $u_i\widehat{x}$ of weight $a_i$, $v_ix$ of weight $-b_i$ and $v_i\widehat{x}$ of weight $b_i+\epsilon$.  
	We now take two copies $(H,w_\rho)$ on vertex set $\{h,h_1,\dots,h_5\}$ 
	and  $(\widehat{H},\widehat{w}_\rho)$ on vertex set $\{\widehat{h},\widehat{h}_1,\dots,\widehat{h}_5\}$ 
	of the graph $(H,w_\rho)$ and $(9s+s_A-\delta)$-attach $(H,w_\rho,h)$ at $\{x\}$ and $(10s-\delta)$-attach $(\widehat{H},\widehat{w}_\rho,\widehat{h})$ at $\{\widehat{x}\}$. Additionally, we add an edge $h\widehat{h}$ of weight $\rho$ and $v_ih$, $v_i\widehat{h}$ of weight $\rho$ for every $i\in [n]$. 
	For an illustration of the construction 
	see Figure~\ref{fig:EAPartition-cs}. 
	Note that $G$ has vertex cover number $12$ which is witnessed by the vertex cover $\{x,\widehat{x},h_1,h_2,h_4,h_5,h,\widehat{h}_1,\widehat{h}_2,\widehat{h}_4,\widehat{h}_5,\widehat{h}\}$. In the following we argue that $(A,B)$ is a YES-instance of \textsc{$\exists\forall$-partition} if and only if $(G,w)$ admits a core stable partition.\\
	
	First assume that $(A,B)$ is a YES-instance of \textsc{$\exists\forall$-partition} and let $A'\subseteq A$ be a subset such that $\sum_{c\in A'\cup B'}c \not=s/2$ for every $B'\subseteq B$. We define a partition $\mathcal{P}$ of $V(G)$ as follows. The partition $\mathcal{P}$ contains partitions $\mathcal{P}_{(H,w_\rho)}$ and $\mathcal{P}_{(\widehat{H},\widehat{w}_\rho)}$ defined in the previous section. 
	Furthermore, $\mathcal{P}$ contains the sets $\{h,x\}\cup \{u_i:a_i\in A'\}$, $\{\widehat{h},\widehat{x}\}\cup \{u_i:a_i\notin A'\}$ and $\{v_i\}$ for every $i\in [n]$. We claim that $\mathcal{P}$ is core stable and argue that there is no blocking coalition.
	
	Towards a contradiction, assume that  $X$ is a blocking coalition of $\mathcal{P}$. 
	By Lemma~\ref{lem:auxiliaryGadget} \ref{prop:gadgetP3} we know that $X\subseteq \{x,\widehat{x}\}\cup\{u_i,v_i:i\in [n]\}$. We first argue that $\widehat{x}$ must be in $X$. To this end, observe that if any vertex $v_i\in X$, then $\widehat{x}\in X$ as $v_i$ has utility $0$ in $\mathcal{P}$ and the only edge incident to $v_i$ with positive weight is $\widehat{x}v_i$. Additionally, if any vertex $u_i\in X$, then $\widehat{x}\in X$ as $u_i$ has utility $a_i$ in $\mathcal{P}$ and the only edges incident to $u_i$ of positive weight are $xu_i$ and $\widehat{x}u_i$ which both have weight $a_i$. Lastly, if $x\in X$, then $\widehat{x}\in X$ as $x$ has utility at least $9s+s_A-\delta$ in $\mathcal{P}$ and the sum of all positive weights of edges incident to $x$ is $s$ if we exclude  $x\widehat{x}$ and $hx$. Note that $h$ does not join $X$. Since $X\not=\emptyset$, we conclude that $\widehat{x}\in X$. This implies that $x\in X$ as $x$ has utility at least $10s-\delta$ in $\mathcal{P}$ and the sum of all positive weights of edges incident to $x$ if we exclude the edges $x\widehat{x}$ and $\widehat{h}\widehat{x}$ is $s_A$.

	We set $B'=\{b_i: v_i\in X\}$ and claim that $\sum_{c\in A'\cup B'}c =s/2$. First note that $x$ has utility $9s+s_A-\delta+\sum_{a\in A'}a$ in $\mathcal{P}$ while $x$ has utility $9.5 s+\sum_{u_i\in X}a_i-\sum_{b\in B'}b$ in $X$. Since $X$ is a blocking coalition we get that
	$$9s+s_A-\delta+\sum_{a\in A'}a<9.5 s+\sum_{u_i\in X}a_i-\sum_{b\in B'}b\leq 9.5 s+s_A-\sum_{b\in B'}b.$$
	Hence, $s/2\geq\sum_{c\in A'\cup B'}c$ since $\delta<\min_{i\in [n]}\{a_i,b_i\}$. Furthermore, $\widehat{x}$ has utility $10s-\delta+\sum_{a\notin A'}a$ in $\mathcal{P}$ while $\widehat{x}$ has utility $9.5 s+\sum_{u_i\in X}a_i+\sum_{b\in B'}(b+\epsilon)$ in $X$. As $X$ is a blocking coalition and $\sum_{a\notin A'}a=s_A-\sum_{a\in A'}a$ we get that
	$$10s-\delta+s_A-\sum_{a\in A'}a<9.5 s+\sum_{u_i\in X}a_i+\sum_{b\in B'}(b+\epsilon)\leq 9.5 s+s_A+\sum_{b\in B'}b+n\epsilon.$$
	Hence, $s/2\leq\sum_{c\in A'\cup B'}c$ since $\delta<\min_{i\in [n]}\{a_i,b_i\}$ and $n\epsilon<\delta$. Combined we get that $\sum_{c\in A'\cup B'}c=s/2$ which contradicts the choice of $A'$.
	\\
	
	Now assume that $(G,w)$ is a YES-instance of \textsc{CS} and let $\mathcal{P}$ be a connected, core stable partition of $G$.  We use the following. We first argue that
	$\{v_i\}$ is a part in $\mathcal{P}$ for every $i\in [n]$.
	First observe that $v_i$ cannot be in the same part as $\widehat{x}$ as by Lemma~\ref{lem:auxiliaryGadget} \ref{prop:gadgetP1} this part also contains $\widehat{h}$ and $w(\widehat{h},v_i)=\rho$. Furthermore, every edge incident to $v_i$ excluding $\widehat{x}v_i$ has negative weight. Hence, if $\{v_i\}\notin \mathcal{P}$ then $\{v_i\}$ would be a blocking coalition.
	
	By Lemma~\ref{lem:auxiliaryGadget} \ref{prop:gadgetP1} there are parts $P,\widehat{P}$ of $\mathcal{P}$ such that $\{h,x\}\subseteq P$ and $\{\widehat{h},\widehat{x}\}\subseteq \widehat{P}$. Additionally, $P\not=\widehat{P}$ as $w(h\widehat{h})=\rho$.
	We set $A'=\{a_i: u_i\in P\}$. In the following we argue that $\sum_{c\in A'\cup B'}c\not=s/2$ for every $B'\subseteq B$. Towards a contradiction, assume that there is $B'\subseteq B$ such that $\sum_{c\in A'\cup B'}c=s/2$. We define $X$ to be the set of vertices $\{x,\widehat{x}\}\cup\{u_i:i\in [n]\}\cup\{v_i: b_i\in B'\}$ and claim that this is a blocking coalition. %We argue that every vertex in $X$ has strictly better utility in $X$ than in $\mathcal{P}$. 
	First note that the utility of every vertex $v_i$ is $0$ in $\mathcal{P}$ as $\{v_i\}\in \mathcal{P}$. On the other hand, the utility of $v_i\in X$ is $\epsilon>0$ since $x,\widehat{x}\in X$ but $h,\widehat{h}\notin X$. Similarly, the utility of every $u_i$ is at most $a_i$ in $\mathcal{P}$ as $x$ and $\widehat{x}$ are contained in different parts by \cref{lem:auxiliaryGadget} \ref{prop:gadgetP1} ($h$ and $\widehat{h}$ are in different part) while it is $2a_i$ in $X$ since $x,\widehat{x}\in X$. The utility of $x$ in $\mathcal{P}$ is $9s+s_A-\delta+\sum_{a\in A'}a$ which is equal to $9.5s+s_A-\delta-\sum_{b\in B'}b$. On the other hand, the utility of $x$ is $9.5s+s_A-\sum_{b\in B'}b>9.5s+s_A-\delta-\sum_{b\in B'}b$ in $X$. Finally, $\widehat{x}$ has utility at most $10s-\delta+\sum_{a\notin A'}a=10s-\delta+s_A-\sum_{a\in A'}a$ in $\mathcal{P}$. Furthermore, the utility of $\widehat{x}$ in $X$ is $9.5s+s_A+\sum_{b\in B'}(b+\epsilon)\ge 10s+s_A-\sum_{a\in A'}a$ by the assumption that $\sum_{c\in A'\cup B'}c=s/2$. We conclude that $X$ is a blocking coalition which contradicts the choice of $\mathcal{P}$.
\end{proof}

\subsection{Core stability parameterized by maximum degree and treewidth}% - improved complexity bounds}
We first prove the algorithmic result of the section.

\begin{theorem}\label{thm:doubleExpAlgoCS}
\textsc{Core Stability} can be solved in time $2^{2^{{O}(\Delta\tw)}}n^{O(1)}$.
\end{theorem}

Before we prove Theorem~\ref{thm:doubleExpAlgoCS}, let us sketch our high-level
strategy. Given an instance of \textsc{Core Stability}, we want to produce an
equivalent instance $\phi$ of $\exists\forall$-SAT, such that $\phi$ has
treewidth roughly $\Delta\tw$, where $\Delta,\tw$ are the maximum degree and
treewidth of the original instance.  We could then use the known
(double-exponential) algorithm for $\exists\forall$-SAT (\cite{Chen04}) to
solve our problem.  Intuitively, we would then attempt to use the existential
part of $\phi$ to encode the ``there exists a partition'' part of the problem,
and the universal part to encode the ``all blocking coalitions fail'' part.

Fundamentally, this strategy is sound and works in a relatively straightforward
way for the universal part: we use a boolean variable for each vertex (to
encode whether it belongs in the blocking coalition) and to check that a
blocking coalition fails for a vertex $v$ we need to place a constraint on 
$v$ and all its (at most $\Delta$) neighbors. This means that a tree
decomposition of $\phi$ should be constructible from a tree decomposition of
the square of the original graph, which would have width at most $\Delta\tw$.

Where we run into some more difficulties, however, is in encoding the
existential part. Intuitively, this is because encoding the partition of the
vertices of a bag into coalitions requires a super-linear number of bits, hence
it is not sufficient to define a variable for each vertex. Indeed, to simplify
things, we define a variable for each \emph{pair} of vertices that appear
together in a bag, encoding whether they are together in a coalition. This
means that the treewidth of the formula $\phi$ we construct is in fact not
$O(\Delta\tw)$ but actually can only be upper-bounded by $O(\tw^2+\Delta\tw)$.

Nevertheless, we insist on obtaining an algorithm that is double-exponential
``only'' in $\Delta\tw$, and not in $\tw^2$.  In order to circumvent our
difficulty we observe that the term that is super-linear in treewidth only
depends on existentially quantified variables.  Thankfully, we manage to show,
via an argument that is more careful than that of \cite{Chen04}, that
$\exists\forall$-SAT has a complexity that only needs to be double-exponential
in the number of \emph{universally quantified} variables of each bag
(\cref{prop:chen-better}).  Using this, we are able to show that the second
exponent of the running time is ``only'' $O(\Delta\tw)$, which as we show later
is optimal, even when $\Delta=O(1)$, under the ETH.

Let us now give some more details. We first recall that $\exists\forall$-SAT is
a variant of the SAT problem where we aim to decide the satisfiability of a
given quantified Boolean formula (QBF) $\phi$  which is of the form $\exists
x_1 \dots \exists x_k \forall y_1\dots \forall y_\ell \psi$ where $\psi$ is a
DNF formula on variables $x_1,\dots,x_k,y_1,\dots,y_\ell$. Two common ways of
associating structure of satisfiability problems is to consider the primal or
incidence graph of the formula. The primal graph of a formula $\phi$ (in CNF or
DNF) is a graph on the set of variables of $\phi$ where two variables are
adjacent if they appear in the same clause. Similarly, the incidence graph is a
bipartite graph on the set of variables and clauses of $\phi$ where a variable
is adjacent to all clauses it appears in.  For convenience, we use a variant of
$\exists\forall$-SAT. We say that a QBF $\phi$ is in  $\ECNFADNF$ if $\phi$ can
be written as 

$$\phi=\exists x_1 \dots \exists x_k \bigwedge_{i=1}^{k'}d_i \hspace{8pt} 
\forall y_1\dots \forall y_\ell \bigvee_{i=1}^{\ell'}c_i$$

for some $k,k',\ell,\ell'\in \mathbb{N}$ where $d_i$ are disjunctive clause
over variables $x_1,\dots,x_k$ containing at most $3$ literals per clause and
$c_i$ are conjunctive clauses over variables $x_1,\dots,x_k,y_1,\dots,y_\ell$.
We present an algorithm for $\ECNFADNF$ in several steps. First, we give an
algorithm with a more careful running time than that of \cite{Chen04}.

\begin{proposition}\label{prop:chen-better} \ifthenelse{\boolean{short}}{\textup{($\star$)}}{}
There is an algorithm that takes as input an instance of $\exists\forall$-SAT
$\exists x\forall y \phi(x,y)$, where $x,y$ are tuples of boolean variables,
$\phi$ is in 3-DNF, and a tree decomposition of the primal graph of $\phi$
where each bag contains at most $t_\exists$ existentially quantified variables
and at most $t_\forall$ universally quantified variables and decides if the
input is satisfiable in time $2^{O(t_\exists+2^{t_\forall})}|\phi|^{O(1)}$.  

\end{proposition}
\begin{toappendix}
\begin{proof}[Proof\ifthenelse{\boolean{short}}{ of \cref{prop:chen-better}}{}]
	
	Our strategy is to reduce the problem to CNF satisfiability, which can be
	solved in time single-exponential in the primal treewidth. In particular, our
	goal is to produce a CNF formula $\psi$ with treewidth
	$O(t_\exists+2^{t_\forall})$, so that deciding if $\psi$ is satisfiable will be
	equivalent to the original problem. Since CNF satisfiability on instances with
	$n$ variables and primal treewidth $t$ can be solved in $2^{O(t)}n^{O(1)}$, and
	since our reduction will run in time $2^{O(t_\exists+t_\forall)}|\phi|^{O(1)}$,
	we will obtain an algorithm with the promised running time.
	
	More precisely, starting from $\phi$, which is a formula in 3-DNF, we would
	like to construct a CNF formula $\psi$, using the same existential variables
	$x$, as well as some new existential variables $z$, so that the following holds
	for each truth assignment to the variables $x$: there exists an assignment to
	$z$ such that $\psi(x,z)$ is true if and only if for all assignments to $y$ we
	have that $\phi(x,y)$ is true. In symbols: $\forall x \left( \exists z
	\psi(x,z) \Leftrightarrow \forall y \phi(x,y) \right)$. It is not hard to see
	that if we prove this equivalence for all assignments $x$, then $\psi$ is
	satisfiable (i.e. $\exists x \exists z \psi(x,z)$) if and only if $\exists x
	\forall y \phi(x,y)$ holds, so the problem reduces to deciding CNF
	satisfiability for $\psi$.
	
	The intuitive idea of the transformation of $\phi$ to $\psi$ is that we want to
	encode in a CNF formula the execution of an algorithm which, having fixed the
	assignment to $x$, decides if there exists an assignment to $y$ to make
	$\phi(x,y)$ false. This algorithm would work using the standard DP methodology
	for treewidth, meaning it would store in each bag, for each assignment to the
	variables of $y$ contained in this bag, a bit of information indicating whether
	this assignment can be extended in a way that renders the formula we have seen
	so far false.
	
	More precisely, the transformation is the following: for each bag $B$ of the
	tree decomposition of the primal graph of $\phi$, for each assignment $\sigma$
	to the (at most $t_\forall$) universal variables of $B$, we construct two new
	variables $z_{B,\sigma}$ and $w_{B,\sigma}$. In order to explain the intended
	meaning of these variables, fix an assignment to $x$. Then, $z_{B,\sigma}$ is
	meant to be set to true if and only if there is an assignment to the $y$
	variables contained in the sub-tree rooted at $B$, consistent with $\sigma$,
	such that every conjunctive clause of $\phi$ contained in this sub-tree is made
	false by the joint assignment to $x$ and $y$. Furthermore, $w_{B,\sigma}$ is
	meant to be true if all conjunctive clauses of $\phi$ whose variables are fully
	contained in $B$ are set to false by the joint assignment to $x$ and the
	assignment $\sigma$.
	
	Let us now describe the clauses of $\psi$. Assume, to simplify things, that we
	are given a nice tree decomposition of the primal graph of $\phi$, such that
	the leaves and the root of the decomposition are empty bags.  We construct the
	following clauses:
	
	\begin{enumerate}
		
		\item For each bag $B$, each assignment $\sigma_1$ to the existential variables
		of $B$, and each assignment $\sigma_2$ to the universal variables of $B$, we
		check if there exists a clause of $\phi$ which is fully contained in $B$ and
		which is set to true by the joint assignments $\sigma_1, \sigma_2$. If this is
		the case, we add to $\psi$ a clause which is falsified exactly when the
		existential variables of $B$ take assignment $\sigma_1$, and add to this clause
		the literal $\neg w_{B,\sigma_2}$. If this is not the case (that is, all
		clauses contained in $B$ are set to false by the joint assignment) we add to
		$\psi$ a clause falsified when the existential variables take assignment
		$\sigma_1$, and add to this clause the literal $w_{B,\sigma_2}$.
		
		\item For each leaf bag $B$, since $B$ is empty, there exists a unique
		(vacuous) assignment $\sigma_{\emptyset}$ to its universal variables. We add to
		$\psi$ the unit clause $(z_{B,\sigma_{\emptyset}})$.
		
		\item For each bag $B$ with a single child $B'$, for each assignment $\sigma$
		to the universal variables of $B$, if there is a unique assignment $\sigma'$ to
		the universal variables of $B'$, such that $\sigma,\sigma'$ agree on their
		common variables (this happens if the universal variables of $B'$ are all
		contained in $B$), we add to $\psi$ clauses that implement the constraint
		$z_{B,\sigma}\leftrightarrow (w_{B,\sigma}\land z_{B',\sigma'})$.  Otherwise,
		since the decomposition is nice, $B'$ contains a single universal variable not
		appearing in $B$ and there are two assignments $\sigma', \sigma''$ to universal
		variables of $B'$ consistent with $\sigma$. We then add to $\psi$ clauses
		implementing the constraint $z_{B,\sigma}\leftrightarrow (w_{B,\sigma}\land
		(z_{B',\sigma'}\lor z_{B',\sigma''}))$.
		
		\item Similarly, for each join bag $B$, with children $B', B''$, for each
		assignment $\sigma$ to the universal variables of $B$, we add to $\psi$ clauses
		that implement the constraint $z_{B,\sigma} \leftrightarrow (w_{B,\sigma}\land
		z_{B',\sigma}\land z_{B'',\sigma})$.
		
		\item In the root bag $B_r$, since $B_r$ is empty, again there exists a unique
		vacuous assignment $\sigma_{\emptyset}$. We add to $\psi$ the unit clause
		$(\neg z_{B,\sigma_{\emptyset}})$.
		
	\end{enumerate}
	
	This completes the construction. The intuition behind the construction is that,
	if we fix the assignment to the $x$ variables, then in order to satisfy $\psi$,
	for each $B,\sigma$, the variable $w_{B,\sigma}$ must correctly encode whether
	$B$ contains a clause that is set to true by the assignment to the $x$
	variables in conjunction with $\sigma$ (in particular, $w_{B,\sigma}$ is true
	if and only if all clauses contained in $B$ are false).  The variable
	$z_{B,\sigma}$ is meant to check if there exists some assignment to the $y$
	variables extending $\sigma$ that makes all clauses of $\phi$ contained in the
	sub-tree rooted at $y$ false; this is vacuously true for a leaf bag (as no
	clauses are contained in the sub-tree);  in other bags this is true if and only
	if all clauses contained in $B$ are false (encoded by $w_{B,\sigma}$) and an
	assignment consistent with $\sigma$ sets all clauses contained in the sub-trees
	rooted at children of $B$ to false. Finally, the given assignment to $x$ has
	the property that $\forall y. \phi(x,y)$ if and only if in the root it is not
	possible to find an assignment to the $y$ variables that makes all clauses
	false, hence the last clause of $\psi$.  Correctness of the construction can be
	argued by induction using the above intuition in a standard way, so we skip the
	details.
	
	The construction above can be executed in time
	$2^{t_\exists+t_\forall}|\phi|^{O(1)}$. In order to construct a tree
	decomposition of $\psi$ we keep the decomposition of $\phi$ and remove all
	universal variables. Then, in each bag $B$ we place all the at most $2\cdot
	2^{t_\forall}$ variables $z_{B,\sigma}, w_{B,\sigma}$ constructed for this bag.
	We then add to each bag $B$ all the variables contained in one of its children.
	Since the decomposition is nice, so each node has at most two children, this
	will at most triple the size of each bag, making all bags contain at most
	$3t_\exists + 6\cdot 2^{t_\forall}$ variables. Observe that now all clauses of
	$\psi$ are fully contained in some bag. We now have a CNF formula $\psi$ and a
	tree decomposition of its primal graph of width $O(t_\exists+2^{t_\forall})$
	such that deciding if $\psi$ is satisfiable resolves our original problem.
	Using a standard DP algorithm now completes the proof.  \end{proof}
\end{toappendix}

\begin{proposition}\label{proposition:existCNFforallDNF} \ifthenelse{\boolean{short}}{\textup{($\star$)}}{}
There is an algorithm that takes as input an $\ECNFADNF$-SAT instance $\phi$
and a tree decomposition of its incidence graph of width $t$ that contains at
most $t_\forall$ universally quantified variables in each bag and at most $2$
clauses in each bag, and decides $\phi$ in time $2^{O(t2^{t_\forall})}|\phi|$.

\end{proposition}

\begin{toappendix}
\begin{proof}[Proof\ifthenelse{\boolean{short}}{ of \cref{proposition:existCNFforallDNF}}{}]
	
	We reduce $\ECNFADNF$-SAT to $\exists \forall$-SAT and then \ifthenelse{\boolean{short}}{make use of}{use}
	\cref{prop:chen-better}.  Let $\phi$ be a formula in $\ECNFADNF$. Hence, $\phi$ can be written as
	$\phi=\exists x_1 \dots \exists x_k \bigwedge_{i=1}^{k'}d_i \hspace{8pt}
	\forall y_1\dots \forall y_\ell \bigvee_{i=1}^{\ell'}c_i$ for some
	$k,k',\ell,\ell'\in \mathbb{N}$ where $d_i$ are disjunctive clauses containing
	at most $3$ literals and $c_i$ are conjunctive clauses. We construct an
	instance $\psi$ of $\exists\forall$-SAT as follows. The formula $\psi$ contains
	existentially quantified variable $x_1,\dots,x_k$ and universally quantified
	variables $y_1,\dots,y_\ell, y_{d_1},\dots, y_{d_{k'}}, y_C$.  Furthermore,
	$\psi$ contains one clause $(y_{d_i},\lambda)$ for every $i\in [k']$ and every
	literal $\lambda$ contained in the clause $d_i$. Additionally, $\psi$ contains
	the clause $y_C\land c_i$ for every $i\in [\ell']$. Finally, $\psi$ contains
	the clause $(\lnot y_{d_1}\land \dots \land \lnot y_{d_{k'}}\land \lnot y_C)$.
	We now argue that $\phi$ is satisfiable if and only if $\psi$ is satisfiable.\\
	
	First assume that $\phi$ is satisfiable and let $\alpha_X:\{x_1,\dots,x_k\}\rightarrow \{0,1\}$ be an assignment which satisfies $\bigvee_{i=1}^{k'}d_i$ and for every assignment $\alpha_Y:\{y_1,\dots,y_\ell\}\rightarrow \{0,1\}$ the formula $\bigwedge_{i=1}^{\ell'}c_i$ is satisfied. Let $\alpha_Y:\{y_1,\dots,y_\ell, y_{d_1},\dots, y_{d_{k'}}, y_C\}\rightarrow \{0,1\}$ be any assignment. 
	
	First assume that $\alpha_Y(y_{d_i})=1$ for some $i\in [k']$. Since the disjunctive clause $d_i$ is satisfied under the assignment $\alpha_X$ by assumption, there is a literal $\lambda$ contained in $d_i$ such that $\lambda$ evaluates to $1$ under assignment $\alpha_X$. Hence, the clause $(y_{d_i},\lambda)$ is satisfied under the assignments $\alpha_X$ and $\alpha_Y$.
	
	Now assume that $\alpha_Y(y_{d_i})=0$ for every $i\in [k']$ and $\alpha_Y(y_C)=0$. In this case the clause $(\lnot y_{d_1}\land \dots \land \lnot y_{d_{k'}}\land \lnot y_C)$ is satisfied under the assignments $\alpha_X$ and $\alpha_Y$.
	
	Finally, assume that $\alpha_Y(y_{d_i})=0$ for every $i\in [k']$ and $\alpha_Y(y_C)=1$. Observe that by assumption $\bigwedge_{i=1}^{\ell'}c_i$ is satisfied under assignments $\alpha_X$ and the restriction of $\alpha_Y$ to the set $\{y_1,\dots,y_\ell\}$. Let $i\in [\ell']$ be the index of the satisfied clause. In this case, the clause $y_C\land c_i$ must be satisfied under the assignment $\alpha_X$ and $\alpha_Y$. Hence, we have argued that $\psi$ is satisfiable.\\
	
	On the other hand, assume that $\psi$ is satisfiable and let $\alpha:\{x_1,\dots,x_k\}\rightarrow \{0,1\}$ is an assignment such that for every assignment $\alpha_Y:\{y_1,\dots,y_\ell, y_{d_1},\dots, y_{d_{k'}}, y_C\}\rightarrow \{0,1\}$ either one of clauses $(y_{d_i},\lambda)$, $i\in [k']$ and $\lambda$ a literal contained in $d_i$ or one of the clauses $y_C\land c_i$, $i\in [\ell']$ or the clause $(\lnot y_{d_1}\land \dots \land \lnot y_{d_{k'}}\land \lnot y_C)$ is satisfied under assignment $\alpha_X$ and $\alpha_Y$. To argue that $\psi$ is satisfiable, let $\alpha_Y:\{y_1,\dots,y_\ell\}\rightarrow \{0,1\}$ be any assignment. We have to argue that every clause $d_i$, $i\in [k']$ is satisfied and at least one clause $d_i$, $i\in [\ell']$ is satisfied.
	
	Fix some $i\in [k']$ and define an assignment $\tilde{\alpha}_Y:\{y_1,\dots,y_\ell, y_{d_1},\dots, y_{d_{k'}}, y_C\}\rightarrow \{0,1\}$ by $\tilde{\alpha}(y_j)=\alpha_Y(y_j)$ for $j\in [\ell]$, $\tilde{\alpha}_Y(y_{d_i})=1$, $\tilde{\alpha}_Y(y_{d_j})=0$ for $j\in [k']$, $j\not=i$ and $\tilde{\alpha}_Y(y_C)=0$. By construction, the only clauses of $\psi$ that can be satisfied under assignments $\alpha_X$ and $\tilde{\alpha}_Y$ are the clauses $(d_i\land \lambda)$ for every literal $\lambda$ in $d_i$. Hence, one literal of $d_i$ evaluates to $1$ under assignment $\alpha_X$ and therefore $d_i$ is satisfied under assignment $\alpha_X$.
	
	Now consider the assignment $\tilde{\alpha}_Y:\{y_1,\dots,y_\ell, y_{d_1},\dots, y_{d_{k'}}, y_C\}\rightarrow \{0,1\}$ where $\tilde{\alpha}(y_j)=\alpha_Y(y_j)$ for $j\in [\ell]$, $\tilde{\alpha}_Y(y_{d_j})=0$ for every $j\in [k']$ and $\tilde{\alpha}_Y(y_C)=1$. Hence, the only clauses of $\psi$ that can be satisfied under assignments $\alpha_X$ and $\tilde{\alpha}_Y$ are the clauses $(y_C\land c_i)$, $i\in [\ell']$. Let $i\in [\ell']$ be the index for which $(y_C\land c_i)$ is satisfied. But in this case the clause $c_i$ is satisfied concluding the argument that $\phi$ is satisfiable.\\
	
	%    We now bound the incidence treewidth of $\psi$ in terms of the incidence treewidth of $\phi$. Let $G^\phi$ be the incidence graph of $\phi$, $G^\psi$ the incidence graph of $\psi$ and $(T,\beta^\phi)$ a tree decomposition of $G^\phi$. Here, we denote the bag of a node $t\in T$ by $\beta^\phi(t)$. We construct a tree decomposition $(T,\beta^\psi)$ of $G^\psi$ as follows. We set 
	%    \begin{align*}
		%        \beta^\psi(t)=&\Big(\beta^\phi(t)\cap \{x_1,\dots,x_k,y_1,\dots,y_\ell\}\Big)\cup \{y_{d}:d\in \beta^\phi(t)\}\cup \{y_C\}
		%        \\&\cup  \Big\{(y_{d}\land\lambda):d\in \beta^\phi(t),\lambda \text{ a literal in }d\Big\}\cup \Big\{(y_C\land c): c\in \beta^\phi(t)\Big\}
		%        \\&\cup \Big\{(\lnot y_{d_1}\land \dots \land \lnot y_{d_{k'}}\land \lnot y_C)\Big\}.
		%    \end{align*} 
	%    It is straightforward to verify that $(T,\beta^\psi)$ is a tree decomposition of $G^\psi$. Furthermore, the width of $(T,\beta^\psi)$ is $4\tw(G^\psi)+2$ as for every $y_d$ the clause $d$ contains at most three literals. As constructing $\psi$ takes linear time, combining the reduction with the $2^{2^{\tw(G^\psi)}}$ algorithm from \cite{Chen04} and the algorithm from \cite{LampisMM18} for converting between primal and incidence treewidth, yields the statement. 
	
	We now consider the incidence graph of $\psi$ and construct a tree
	decomposition of that graph starting from a decomposition of the incidence
	graph of $\phi$. First, we maintain all common variables, and for each
	conjunctive clause $c$ in a bag, we replace it with the clause $(y_C\land c)$
	which has replaced it in $\psi$; we add $y_C$ and the clause $(\neg
	y_{d_1}\land \ldots\land y_{d_{k'}\land \neg y_C})$ to all bags; for each
	disjunctive clause $d$ appearing in a bag we add the three new clauses
	$(y_d\land \lambda)$, where $\lambda$ is a literal of $d$, to the bag. Observe
	that we have a valid decomposition of the new incidence graph, with width
	$t'=O(t)$, such that each bag contains at most $7$ clauses, and at most
	$t_\forall+1$ universal variables. 
	
	We are now almost done, except that we need to obtain a decomposition of the
	primal graph of $\psi$ to invoke \cref{prop:chen-better} and also to convert
	$\psi$ into 3-DNF. This, however, can be done using an algorithm of
	\cite{LampisMM18}. In particular, we convert $\psi$ to 3-DNF in the standard
	way: as long as there exists a clause that contains at least $4$ literals, say
	$(\lambda_1\land \lambda_2\land \ldots\land \lambda_r)$, we introduce a new
	universally quantified variable $z$ and replace it with the two clauses
	$(\lambda_1\land \lambda_2\land z)$ and $(\neg z\land
	\lambda_3\land\ldots\land\lambda_r)$. The difference, however, is that we
	select the two literals to place in the new clause in a way that keeps the
	treewidth of the new formula under control. In particular, we find the lowest
	bag $B$ in the decomposition that contains the current clause such that the
	sub-tree rooted at $B$ contains two variables of the clause and set
	$\lambda_1,\lambda_2$ to be the literals involving these two variables.  It is
	now not hard to see that if we execute this exhaustively, the new decomposition
	we construct will have $O(1)$ clauses per bag, $t_\forall+O(1)$ universally
	quantified variables, and width $O(t)$. Now, replace each clause in the
	decomposition by the (at most $3$) variables it contains, and we get a
	decomposition of the primal graph, on which we can invoke
	\cref{prop:chen-better}.  \end{proof}
\end{toappendix}

\begin{proof}[Proof of Theorem~\ref{thm:doubleExpAlgoCS}]
%We obtain an algorithm through a reduction to $\ECNFADNF$-SAT.  
We prove Theorem~\ref{thm:doubleExpAlgoCS} by a reduction to $\ECNFADNF$-SAT. 
Let $(G,w)$ be an instance of  \textsc{CS} and $(T,\beta)$ be a rooted tree decomposition of $G$. Here, we denote the bag of a node $t\in T$ by $\beta(t)$. 
First we let $(G',w')$ be the graph obtained from $(G,w)$ by adding edges $uv$ of weight $0$ for every pair of vertices $u,v$ appearing in a bag together. It is straightforward to see that $(G,w)$ is a YES-instance of \textsc{CS} if and only if $(G',w')$ is a YES-instance of \textsc{CS}. Additionally, $G'$ is chordal and $(T,\beta)$ is a tree decomposition of $G'$.

%We construct an instance $\phi$ of $\exists\forall$-SAT.
We construct an instance $\phi$ of $\ECNFADNF$-SAT. 
We introduce a variable $x_e$ for every $e\in E(G')$ and a variable $y_u$ for every vertex $u\in V(G')$. 
An assignment $\alpha_X:\{x_e:e\in E(G')\}\rightarrow \{0,1\}$ represents a subset of $E(G')$ and an assignment $\alpha_Y: \{y_u:u\in V(G')\}\rightarrow \{0,1\}$ represents a subset of  $V(G')$. Intuitively,  a partition of $V(G')$ corresponds to a set of edges, i.e. by including all edges that are incident to vertices from the same part. % the set of edges represented by an assignment $\alpha_X:\{x_e:e\in E(G')\}\rightarrow \{0,1\}$. 
For every $t\in V(T)$ we define a formula $\phi_t$, which we use to enforce sufficient criteria for the set of edges represented by an assignment $\alpha_X:\{x_e:e\in E(G')\}\rightarrow \{0,1\}$ to correspond to a partition of $V(G')$. Let
\begin{align*}
	\phi_t=&\bigwedge_{u_1,u_2,u_3\in \beta(t)}\big(x_{u_1u_2}\land x_{u_2u_3}\big)\rightarrow x_{u_1u_3}=
	\bigwedge_{u_1,u_2,u_3\in \beta(t)}\Big(\lnot x_{u_1u_2}\lor \lnot x_{u_2u_3}\lor x_{u_1u_3}\Big).
\end{align*}
For a fixed partition $\mathcal{P}$ represented by some assignment $\alpha_X:\{x_e:e\in E(G')\}\rightarrow \{0,1\}$ our formula needs to ensure that no blocking coalitions exist. This is realized by guaranteeing that for each assignment $\alpha_Y: \{y_u:u\in V(G')\}\rightarrow \{0,1\}$ the set corresponding to $\alpha_Y$ is not blocking. Intuitively, the formula $\phi_u$ defined below ensures that the utility of vertex $u$ in $\mathcal{P}$ is at least as large as the utility of $u$ in the coalition represented by $\alpha_Y$. For $u\in V(G')$ let
$$\phi_u=\hspace{-5pt}\bigvee_{N,\tilde{N}\subseteq N_{G'}(u),\atop{\sum_{v\in N}w(uv)\leq\sum_{v\in \tilde{N}}w(uv)}}\hspace{-10pt}\bigg(\bigwedge_{v\in N}y_{v}\land \bigwedge_{v\in N_{G'}(u)\setminus N}\lnot y_{v} \land \bigwedge_{v\in \tilde{N}}x_{uv}\land \bigwedge_{v\in N_{G'}(u)\setminus\tilde{N}}\lnot x_{uv}\bigg).$$
We now define the formula $\phi$ to be 
$$\phi=\exists x_{e_1}\dots \exists x_{e_m}\bigwedge_{t\in V(T)}\phi_t \hspace{10pt}\forall y_{v_1} \dots \forall y_{v_n} \big(\lnot y_1\land \dots \land \lnot y_n \big)\lor \bigvee_{u\in V(G')}\phi_u.$$ 
Observe that $\phi$ is in $\ECNFADNF$.
In the following we prove that $(G',w')$ is a YES-instance of \textsc{CS} if and only if $\phi$ is satisfiable.\\

First assume that $\mathcal{P}$ is a core stable partition of $V(G')$. We define an assignment $\alpha_X:\{x_e:e\in E(G')\}\rightarrow \{0,1\}$ by setting $\alpha_X(x_e)=1$ if $e$ is incident to two vertices residing in the same part of $\mathcal{P}$ and $\alpha_X(x_e)=0$ otherwise. This assignment satisfies $\phi_t$ for every $t\in V(T)$ as for any three vertices $u_1,u_2,u_3\in \beta(t)$ it holds that if $\alpha_X(x_{u_1u_2})=1$ and $\alpha_X(x_{u_2u_3})=1$, then $u_1,u_2$ and $u_3$ must reside in the same part of $\mathcal{P}$ and hence $\alpha_X(x_{u_1u_3})=1$. Furthermore, consider any assignment $\alpha_Y: \{y_u:u\in V(G')\}\rightarrow \{0,1\}$ and let $X=\{u\in V(G'):\alpha_Y(y_u)=1\}$. We have to argue that the formula $\big(\lnot y_1\land \dots \land \lnot y_n \big)\lor \bigvee_{u\in V(G')}\phi_u$ is satisfied under the assignments $\alpha_X$ and $\alpha_Y$. In case that $X=\emptyset$, the clause  $\big(\lnot y_1\land \dots \land \lnot y_n \big)$ is satisfied. On the other hand, if $X\not=\emptyset$, then there must be a vertex $u\in X$ whose utility in $\mathcal{P}$ is at least as large as its utility in $X$ as the set $X$ cannot be a blocking coalition. Hence, for the sets $N=N_{G'}(u)\cap X$ and $\tilde{N}=N_{G'}(u)\cap P$, where $P\in \mathcal{P}$ is the part containing $u$, we have that $\sum_{v\in N}w(uv)\leq\sum_{v\in \tilde{N}}w(uv)$. Therefore, $\phi_u$ contains the clause $\big(\bigwedge_{v\in N}y_{v}\land \bigwedge_{v\in N_{G'}(u)\setminus N}\lnot y_{v} \land \bigwedge_{v\in \tilde{N}}x_{uv}\land \bigwedge_{v\in N_{G'}(u)\setminus\tilde{N}}\lnot x_{uv}\big)$ and this clause is satisfied under the assignment $\alpha_X$ and $\alpha_Y$ by choice of $N$ and $\tilde{N}$. This shows that $\phi$ is satisfiable.\\

On the other hand, assume that $\phi$ is satisfiable and let $\alpha_X:\{x_e:e\in E(G')\}\rightarrow \{0,1\}$ be an assignment such that $\bigwedge_{t\in V(T)}\phi_t$ as well as $\forall y_{v_1} \dots \forall y_{v_n} \big(\lnot y_1\land \dots \land \lnot y_n \big)\lor \bigvee_{u\in V(G')}\phi_u$ is satisfied under  $\alpha_X$. We let $E_X=\{e\in E(G'): \alpha_X(x_e)=1\}$ and define a partition $\mathcal{P}$ by letting every part of $\mathcal{P}$ correspond to the vertices of a connected component of the graph $(V(G'),E_X)$. To show that the partition $\mathcal{P}$ is core stable we use the following claim.
\begin{claim}\label{claim:consistencyPartition} \ifthenelse{\boolean{short}}{\textup{($\star$)}}{}
	For every  edge $uv\in E(G')$ it holds that $uv\in E_X$ if and only if $u$ and $v$ are contained in the same part of $\mathcal{P}$.
\end{claim}
\begin{toappendix}
	\begin{proof}[Proof\ifthenelse{\boolean{short}}{ of \cref{claim:consistencyPartition}}{}]
		Since the forward direction holds by definition of $\mathcal{P}$, assume towards a contradiction that $uv\in E(G')$ with $u$ and $v$ being contained in the same part of $\mathcal{P}$ but $uv\notin E_X$. By definition of $\mathcal{P}$, $u$ and $v$ being contained in the same part implies that there is a path from $u$ to $v$ in the graph $(V(G'),E_X)$. Let $Q=(q_1,\dots, q_\ell)$ be a shortest such path. We aim to find an inconsistent triangle, i.e. three vertices $v_1,v_2,v_3$ such that $v_1v_2,v_2v_3\in E_X$ but $v_1v_3\notin E_X$. Assume that $k\leq 2$ is the minimum number such that there is an index $i\in [\ell-k]$ for which $q_iq_{i+k}\in E(G')$. As $G'$ is chordal, we picked $k$ minimum $k=2$. Furthermore, $q_iq_{i+2}\notin E_X$ as otherwise, the path $(q_1,\dots, q_i,q_{i+2},\dots q_\ell)$ is a path in $(V(G'),E_X)$ from $u$ to $v$ which is shorter than $Q$. Hence, $q_i,q_{i+1},q_{i+2}$ forms an inconsistent triangle. We now claim that there is a node $t\in V(T)$ such that $q_i,q_{i+1},q_{i+2}\in \beta(t)$. As $q_iq_{i+1},q_{i+1}q_{i+2},q_iq_{i+2}\in E(G')$ we know that there are nodes $t_1,t_2,t_3\in V(T)$ such that  $\{q_i,q_{i+1}\}\subseteq \beta(b_1),\{q_{i+1},q_{i+2}\}\subseteq \beta(t_2)$ and $\{q_i,q_{i+2}\}\subseteq \beta(t_3)$. Let $t$ be the lowest common ancestor of $t_1,t_2$ and $t_3$ in $T$. As the nodes of $T$ containing a particular vertex of $G'$ in their bags form a subtree of $T$, we directly obtain that $q_i,q_{i+1},q_{i+2}\in \beta(t)$. Hence, $\big(x_{q_iq_{i+1}}\land x_{q_{i+1}q_{i+2}}\big)\rightarrow x_{q_iq_{i+3}}\big)$ is a clause in  $\bigwedge_{t\in V(T)}\phi_t$. Hence, $\big(x_{q_iq_{i+1}}\land x_{q_{i+1}q_{i+2}}\big)\rightarrow x_{q_iq_{i+2}}\big)$ must be satisfied under the assignment $\alpha_X$ which contradicts $q_iq_{i+2}$ not being contained in $E_X$ proving the claim.
	\end{proof}
\end{toappendix}
Towards a contradiction assume that $\mathcal{P}$ is not core stable and let $X$ be a blocking coalition. We define an assignment $\alpha_Y: \{y_u:u\in V(G')\}\rightarrow \{0,1\}$ by setting $\alpha_Y(y_u)=1$ if $u\in X$ and $\alpha_Y(y_u)=0$ otherwise. By assumption, the DNF formula $\big(\lnot y_1\land \dots \land \lnot y_n \big)\lor \bigvee_{u\in V(G')}\phi_u$ is satisfied under assignment $\alpha_X$ and $\alpha_Y$. As $X$ cannot be empty there is at least one $u\in V(G')$ such that $\alpha_Y(y_u)=1$ and hence the clause $\big(\lnot y_1\land \dots \land \lnot y_n \big)$ cannot be satisfied under the assignment $\alpha_Y$ which implies that some clause in $\bigvee_{u\in V(G')}\phi_u$ must be satisfied. Let $u\in V(G')$ and $N,\tilde{N}\subseteq N_{G'}(u)$ with $\sum_{v\in N}w(uv)\leq\sum_{v\in \tilde{N}}w(uv)$ such that the clause $\big(\bigwedge_{v\in N}y_{v}\land \bigwedge_{v\in N_{G'}(u)\setminus N}\lnot y_{v} \land \bigwedge_{v\in \tilde{N}}x_{uv}\land \bigwedge_{v\in N_{G'}(u)\setminus\tilde{N}}\lnot x_{uv}\big)$ is satisfied under assignments $\alpha_X$ and $\alpha_Y$. This implies that $N=N_{G'}(u)\cap X$ and $\tilde{N}=\{v\in V(G'):uv\in E_X\}$. Since by Claim~\ref{claim:consistencyPartition}, $\{v\in V(G'):uv\in E_X\}=N_{G'}(u)\cap P$, where $P$ is the part of $\mathcal{P}$ containing $u$, this implies that the utility of $u$ in $X$ is less or equal to the utility of $u$ in $\mathcal{P}$. As this contradict our assumption that $X$ is a blocking coalition it follows that $\mathcal{P}$ is core stable.

\begin{claim}\label{claim:bdTreewidth} \ifthenelse{\boolean{short}}{\textup{($\star$)}}{}
	The formula $\phi$ has incidence treewidth at most $\Delta(G')\tw(G')+\tw(G')^2+2$. Furthermore, we can construct a decomposition of that width which contains at most $2$ clauses per bag and at most $(\Delta(G')+1)\tw(G')$ universally quantified variables per bag.
\end{claim}
\begin{toappendix}
	\begin{proof}[Proof\ifthenelse{\boolean{short}}{ of \cref{claim:bdTreewidth}}{}]
		Let $G^\phi$ be the incidence graph of $\phi$. For $t\in V(T)$ we denote the subtree of $T$ rooted at $t$ by $T_t$. For every $u\in V(G')$ we let  $\lca(u)$ be the node $t$ such that $u\in \beta(t)$ and $T_t$ contains every node $t'\in V(T)$ for which $u\in \beta(t')$. In other words $\lca(u)$ is the lowest common ancestor of all nodes whose bag contain $u$. We now construct a tree decomposition $(T^\phi,\beta^\phi)$ of $G^\phi$. We obtain $T_\phi$ from $T$ by the following steps. 
		\begin{enumerate}
			\item For every $t\in V(T)$ and every triple $(u_1,u_2,u_3)\in \beta(t)^3$ we introduce a new node $t_{(u_1,u_2,u_3)}$.
			\item For every vertex $u\in V(G)$ and every two sets $N,\tilde{N}\subseteq N_{G'}(u)$ for which $\sum_{v\in N}w(uv)\leq\sum_{v\in \tilde{N}}w(uv)$ we introduce a node $\lca(u)_{(N,\tilde{N})}$.
			\item For the root $r$ of $T$ we attach a path to $r$ containing all nodes $r_{(u_1,u_2,u_3)}$ and all nodes $t_{(N,\tilde{N})}$ in an arbitrary order. 
			\item For every node $t\in V(T)$ apart from the root $r$, we subdivide the edge from $t$ to its ancestor by the nodes $t_{(u_1,u_2,u_3)}$ and the nodes $t_{(N,\tilde{N})}$.
		\end{enumerate}
		We define $\beta^\phi$ in the following way. We set $\beta^\phi(t)=\{y_u:u\in N_{G'}(\beta(t))\cup \beta(t)\}\cup \{x_{uv}:u,v\in \beta(t)\}\cup \{c_\lnot\}$ for every node $t\in V(T)$ where $c_\lnot$ is the clause $(\lnot y_1\land \dots \land \lnot y_n)$. Further, we set $\beta^\phi(t_{(u_1,u_2,u_3)})=\{c_{(u_1,u_2,u_3)}\}\cup \beta ^\phi(t)$ where $c_{(u_1,u_2,u_3)}$ is the clause $(\lnot x_{u_1u_2}\lor \lnot x_{u_2u_3}\lor x_{u_1u_3})$. Similarly, we set $\beta^\phi(t_{(N,\tilde{N})})=\{c_{(N,\tilde{N})}\}\cup \beta ^\phi(t)$ where $c_{(N,\tilde{N})}$ is the clause $(\bigwedge_{v\in N}y_{v}\land \bigwedge_{v\in N_{G'}(u)\setminus N}\lnot y_{v} \land \bigwedge_{v\in \tilde{N}}x_{uv}\land \bigwedge_{v\in N_{G'}(u)\setminus\tilde{N}}\lnot x_{uv})$.
		
		We now argue that the pair $(T^\phi,\beta^\phi)$ is a tree decomposition of $G^\phi$.
		First observe that that every vertex of $G^\phi$ is contained in the bag of some node of $T^\phi$. 
		Next we argue that the two vertices incident to any edge appear in a bag together in $(T^\phi,\beta^\phi)$. Observe that any edge incident to the clause $(\lnot y_1\land \dots \land \lnot y_n)$ must be contained in some bag as $(\lnot y_1\land \dots \land \lnot y_n)$ is contained in every bag. Furthermore, any bag $\beta^\phi(t_{(u_1,u_2,u_3)}$ contains both $y_{u_1},y_{u_2},y_{u_3}$ and $(\lnot x_{u_1u_2}\lor \lnot x_{u_2u_3}\lor x_{u_1u_3})$ and hence every edge incident to  $(\lnot x_{u_1u_2}\lor \lnot x_{u_2u_3}\lor x_{u_1u_3})$ is contained in some bag. Lastly, any bag $\beta(\lca(u)_{N,\tilde{N}})$ contains $y_u$, $y_v$ for every neighbor $v\in N_{G'}(u)$ and the clause  $(\bigwedge_{v\in N}y_{v}\land \bigwedge_{v\in N_{G'}(u)\setminus N}\lnot y_{v} \land \bigwedge_{v\in \tilde{N}}x_{uv}\land \bigwedge_{v\in N_{G'}(u)\setminus\tilde{N}}\lnot x_{uv})$ which implies that every edge incident to $(\bigwedge_{v\in N}y_{v}\land \bigwedge_{v\in N_{G'}(u)\setminus N}\lnot y_{v} \land \bigwedge_{v\in \tilde{N}}x_{uv}\land \bigwedge_{v\in N_{G'}(u)\setminus\tilde{N}}\lnot x_{uv})$ is contained in some bag.\\
		
		Finally, we argue that the subtree induced by nodes containing a fixed vertex $v$ in their bag, is connected for every vertex $v\in V(G^\phi)$. First note that for any clause $c$ of $\phi$ the set $\{t\in V(T^\phi):c\in \beta(t)\}$ trivially induces a subtree of $T$ as any clause is either contained in the bags of all nodes of $T^\phi$ or in the bag of only one node. 
		
		Next, fix $e=uv\in E(G')$ and consider the set $S_e=\{t\in V(T^\phi):x_e\in \beta^\phi(t)\}$. Observe, that for any node $t\in V(T)$ by construction $t\in S_e$ if and only if $u,v\in \beta(t)$. Hence, the set $S_e\cap V(T)$ induces a subtree on $T$ as it is the intersection of two sets $\{t\in V(T):u\in \beta(t)\}$ and $\{t\in V(T): v\in \beta(t)\}$ which induce a subtree on $T$. As the graph induced by $S_e$ on $T$ can be obtained from the subtree of $T$ induced by $S_e\cap V(T)$ by subdividing every edge and attaching a path at the root, we obtain that $S_e$ must induce a tree. 
		
		Finally, fix $u\in V(G)$ and consider the set $S_u=\{t\in V(T):y_u\in \beta^\phi(t)\}$. By construction, a node $t\in V(T)$ is contained in $S_u$ if and only if either $u$ or some neighbor $v\in N_{G'}(u)$ is contained in $\beta(t)$. Hence, the graph induces by $S_u\cap V(T)$ is the union of the trees $G'[\{t\in V(T):v\in \beta(t)\}]$ for all neighbors $v\in N_{G'}(u)\cup\{u\}$. As  $\{t\in V(T):v\in \beta(t)\}$ must intersect the set $\{t\in V(T):u\in \beta(t)\}$  (as $(T,\beta)$ is a tree decomposition of $G'$ and $uv\in E(G')$) the graph induced by $S_u\cap V(T)$ on $T$ is a subtree of $T$. Similarly to the previous case, this implies that $S_u$ induces a subtree of $T^\phi$. \\
		
		To argue about the width of $(T^\phi,\beta^\phi)$ first observe that every bag contains at most two clauses. Excluding clauses, every bag of $(T^\phi,\beta^\phi)$ consists of the set $\{y_u:u\in N_{G'}(\beta(t))\cup \beta(t)\}$ of size $\Delta(G')\tw(G')+\tw(G')$ and $\{x_{uv}:u,v\in \beta(t)\}$ of size $\tw(G')\cdot (\tw(G')-1) $ and hence $(T^\phi,\beta^\phi)$ has width $\Delta(G')\tw(G')+\tw(G')^2+2$. 
	\end{proof}
\end{toappendix}
%$\phi$ contains $\mathcal{O}(\tw(G')^2 n)$ variables ($T$ can be assumed to have linear in $n$ many nodes). 
As $(T,\beta)$, $G'$ and the formula $\phi$ can be computed in time $\mathcal{O}(\tw(G')^2 n)$, combining the reduction with the algorithm from Proposition~\ref{proposition:existCNFforallDNF} yields a %$2^{2^{\mathcal{O}(\Delta(G')\tw(G')+\tw(G')^2+2)}}n\cdot \mathcal{O}(\tw(G')^2 n) \in 
$2^{2^{\mathcal{O}(\Delta\tw)}}n^{O(1)}$ time algorithm.
\end{proof}

Finally, we prove our ETH based lower bound.
\begin{theorem}\label{thm:ETHlowerBoundCS} \ifthenelse{\boolean{short}}{\textup{($\star$)}}{}
Unless the ETH fails, there is no algorithm for \textsc{Core stability} running in time $2^{2^{o(\pw)}}n$ even if $G$ has bounded degree and weights are constant.
\end{theorem}
We give an overview of our construction. Given an instance $\phi$ of $(3,3)$-SAT (each variable appears at most 3 times) we construct a graph $(G,w)$ and show that $\phi$ is satisfyable if and only if $(G,w)$ is core stable. First we tweak the auxiliary gadget from \cref{sec:auxiliaryGadget} slightly such that if we attach $(H, w_\rho)$ at a set of vertices $S$ then $\{h,s\}\in \mathcal{P}$ for some $s\in S$ for every core stable partition $\mathcal{P}$ (details are given in the Appendix). 

We now describe the construction and certain partitions of $V(G)$ which we refer to as  candidate partition.
Every variable $x_i$, $i\in [n]$ of $\phi$ is represented by two vertices $y_i$ and $\lnot y_i$ and we attach an auxiliary gadget at $\{y_i,\lnot y_i\}$. Any candidate partition $\mathcal{P}$ has to contain either $\{h,y_i\}$ or $\{h,\lnot y_i\}$ (but not both) where $h$ is a vertex of the attached gadget. For any other vertex $v$  without an attached auxiliary gadget (excluding vertices of auxiliary gadgets)  $\{v\}$ has to be in any candidate partition. For any vertex $v$ with attached copy $(H,w_\rho)$ the set $\{v,h\}$ has to be in any candidate partition. Using the properties of the auxiliary gadget, we obtain that any core stable partition has to be a candidate partition. 
By construction, there is a correspondence between  assignments and candidate partitions, i.e. $\alpha(x_i)=1$ if and only if $\{h,y_i\}\in \mathcal{P}$ for candidate partition $\mathcal{P}$ and the corresponding assignment $\alpha$. 

Any blocking coalition of a candidate partition allows us to find a clause which is not satisfied under the assignment which corresponds to the candidate partition and vise versa. This is realized as follows. We take $2m+1$ cycles $U^1,\dots,U^m,V^1,\dots,V^m, Z$ of length approximately $3n$ where $2^m\leq 3n$ is the number of clauses of $\phi$. By choosing suitable edge weights, we enforce that any blocking coalition of any candidate partition contains $Z$ and either $U^k$ or $V^k$ (but not both) for every $k\in [m]$. For now, we call any set containing $Z$, $U^k$ or $V^k$ (but not both) for every $k\in [m]$ (and some other vertices we neglect here) a candidate blocking coalition. We number the clauses of $\phi$ in such a way that each candidate blocking coalition corresponds to a clause, i.e. the $j$-th clause corresponds to the candidate blocking coalition $X$ in which $U^k\in X$ if and only if the $k$-th bit of $j$ in binary is $0$. 

Each vertex in $Z$ corresponds to the appearance of a variable.  Assume $z\in Z$ corresponds to the appearance of variable $x_i$ in clause $c_j$.  We connected $z$ to either $y_i$ or $\lnot y_i$ dependent on whether $x_i$ appears negated  in $c_j$. We connect $z$  to either $U^k$ or $V^k$ for every $k\in [m]$ using a special gadget dependent on whether the $k$-th bit of $j$ in binary is $0$ or $1$. The gadget enforces that $z$ obtains a $+1$ towards its total utility in $X$  if and only if the candidate blocking coalition $X$ does not encodes the clause $c_j$. 
By  choice of edge weights, we ensure that vertex $z$ can only be convinced to join a blocking coalition if it either gets $+1$ from its clause selection gadget or  $+1$ from $y_i$ (or $\lnot y_i$, resp.). On the other hand,  $y_i$ ($\lnot y_i$, resp.) can only be convinced to join a blocking coalition if it appears as a singleton in the the partition we are trying to block and hence the corresponding literal is false. In conclusion, for any candidate partition $\mathcal{P}$ there is a blocking coalition $X$ if and only if for the clause corresponding to $X$ each literal is false. Hence, $(G,w)$ is core stable if and only if $\phi$ is satisfiable. 
\begin{toappendix}
\begin{proof}[Proof of \cref{thm:ETHlowerBoundCS}]
	We give a reduction from $(3,3)$-SAT which is the variant of $3$-SAT in which every variable is restricted to appear at most $3$ times. Let $\phi$ be an instance of $(3,3)$-SAT on variables $x_1,\dots,x_n$ (and hence $\phi$ contains at most $3n$ clauses). Without loss of generality, we can assume that no clause contains both $x_i$ and $\lnot x_i$ (because then $\phi$ is trivially satisfiable) and  that $x_i$ as well as $\lnot x_i$ appears in $\phi$ for every $i\in [n]$ (as otherwise we can set the literal that appears to be true). For each $i\in [n]$ we let $b_i\in \{2,3\}$ be the number of appearances of $x_i$ in $\phi$. We denote the clauses of $\phi$ by $c_1,\dots c_{2^m}$, $2^m\leq 3n$.  Here we assumed for easier notation that the number of clauses is a power of $2$.
	
	In the following we construct an instance $(G,w)$ of  \textsc{CS} of bounded degree $20$ and path-width in $\mathcal{O}(\log n)$. The graph $(G,w)$ can be constructed in polynomial time.
	We show that $(G,w)$ admits a core stable partition if and only if $\phi$ is satisfyable. Hence, if we could decide \textsc{CS} in time $2^{2^{o(\pw(G))}}$, then we can solve $(3,3)$-SAT in time $2^{o(n)}$ contradicting the ETH.
	
	We first describe the construction of $(G,w)$. For this, we set $\rho=-20-2m$. %For a more structured presentation of the construction we sometimes $\xi$-neighborhood attach a copy of $(H,w_\rho)$ at a set $S$ without fully having described the neighborhood of $S$. 
	
	\paragraph*{Variable gadget} For each $i \in [n]$ we  introduce vertices $y_i, \lnot y_i, z_{i,1},\dots,z_{i,b_i}$. For each $j\in [b_i]$ we add edge $y_iz_{i,j}$ of weight $1$ if the $j$st clause in which $x_i$ appears contains the literal $x_i$ and we add edge $\lnot y_iz_{i,j}$ of weight $1$ if the $j$st clause in which $x_i$ appears contains the literal $\lnot x_i$. 
	
	For every $i\in [n]$, $j\in [b_i-1]$ we add the edge $z_{i,j}z_{i,j+1}$ of weight $5$. For every $i\in [n-1]$ we add the edge $z_{i,b_i}z_{i+1,1}$ of weight $5$. Hence all  $z_{i,j}$'s are contained in a path in which every edge has weight $5$. Gor an illustration of this gadget see \cref{fig:sub3}.

	\paragraph*{Clause containment gadget}
	For each $i\in [n]$, $j\in [b_i]$, $k\in [0,m+1]$ we introduce a vertex  $s_{i,j}^k$. Additionally, for each $i\in [n]$, $j\in [b_i]$, $k\in [0,m]$ we add a vertex  $\overline{s}_{i,j}^k$. 
	For each $i\in [n]$, $j\in [b_i]$, $k\in [0,m]$ we add edge $s_{i,j}^k\overline{s}_{i,j}^{k}$ of weight 10, $\overline{s}_{i,j}^ks_{i,j}^{k+1}$ of weight $9$ and we add edge $s_{i,j}^{ m+1}z_{i,j}$ of weight $1$. 
	Furthermore, for every $i\in [n]$, $j\in [b_i]$, $k\in [0,m]$ we introduce  vertices  $t_{i,j}^k, u_{i,j}^k$, $v_{i,j}^k$
	and  add edge $s_{i,j}^kt_{i,j}^k$ of weight $-1$. % and $\{u_{i,j}^k,v_{i,j}^k\}$ of weight $\rho$.
	
	For every $i\in [n]$, $j\in [b_i-1]$ $k\in [0,m]$ we add edges $u_{i,j}^ku_{i,j+1}^k$ and $v_{i,j}^kv_{i,j+1}^k$ of weight $5$. For every $i\in [n-1]$, $k\in [0,m]$ we add edges $u_{i,b_i}^ku_{i+1,1}^k$ and $v_{i,b_i}^kv_{i+1,1}^k$ of weight $5$. Hence for every fixed $k$ all  $u_{i,j}^k$'s are contained in a path in which every edge has weight $5$ and equally for the $v_{i,j}^k$'s. 
	
	Fix $i\in [n]$ and $j\in [b_i]$ and assume $c_\ell$ is the clause containing the $j$th appearance of $x_i$. Let $\sum_{k=0}^{m}a_{i,j}^k 2^k$ where $a_{i,j}^k\in \{0,1\}$ be the binary representation of $\ell$. For every $k\in [0,m]$ we add edge $t_{i,j}^ku_{i,j}^k$ of weight $2$ if $a_{i,j}^k=0$ and  edge $t_{i,j}^kv_{i,j}^k$ of weight $2$ if $a_{i,j}^k=1$. For an illustration see \cref{fig:sub1}.
	\begin{figure}
		\centering
		\begin{subfigure}{.48\textwidth}
			\centering
			\includegraphics[scale=0.8]{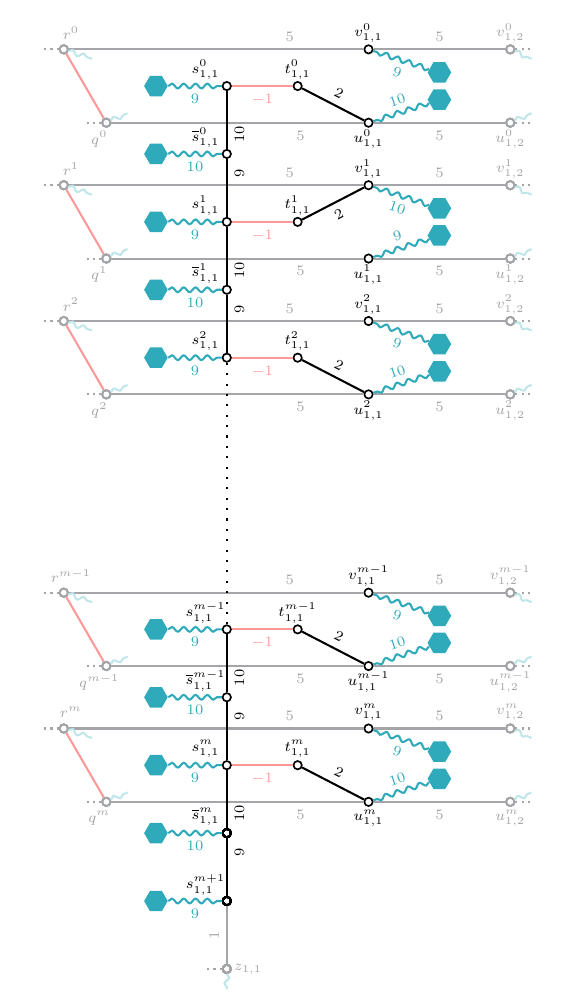}
			\caption{A clause containment gadget encoding that variable $x_1$ is contained in the second clause.}
			\label{fig:sub1}
		\end{subfigure}%
		\hfill
		\begin{subfigure}{.48\textwidth}
			\centering
			\includegraphics[scale=0.8]{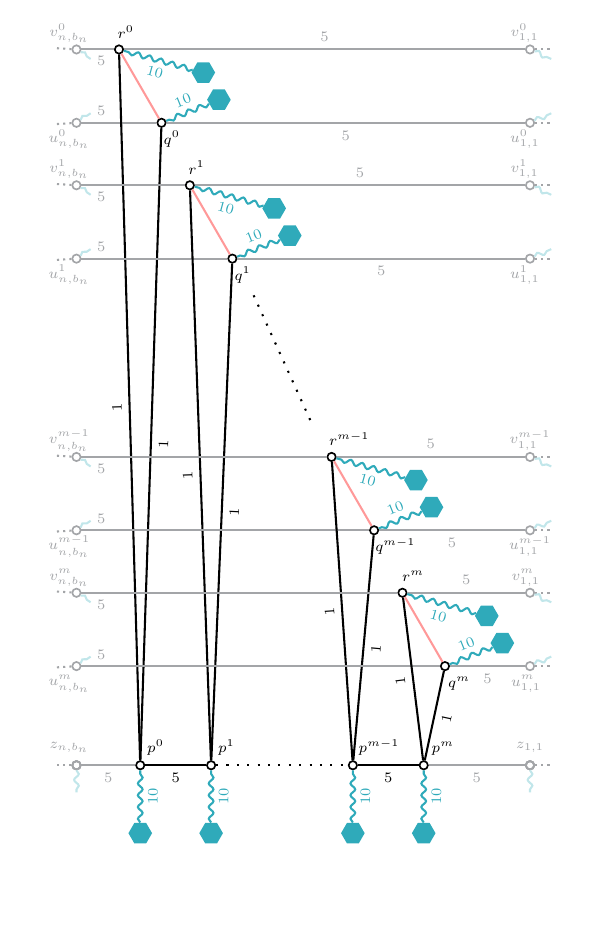}
			\caption{The clause selection gadget.}
			\label{fig:sub2}
		\end{subfigure}
		\begin{subfigure}{.9\textwidth}
			\centering
			\includegraphics[scale=0.8]{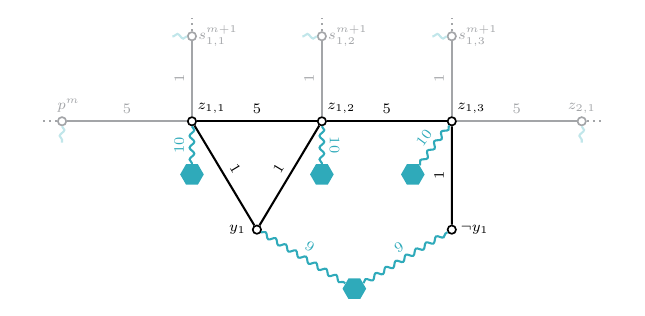}
			\caption{Example of a variable gadget of $x_1$ where $x_1$ appears twice in $\phi$ and $\lnot x_1$ appears once.}
			\label{fig:sub3}
		\end{subfigure}
		
		\caption{The gadgets from the construction in the proof of Theorem~\ref{thm:ETHlowerBoundCS}. Here auxiliary gadgets are depicted my blue hexagons, there attachment at set $S$ is signified by a squiggly edge and red edges without weight are the edges of weight $\rho$.}
		\label{fig:test}
	\end{figure}
	
	%\todo[inline]{ML: The compiler complains about multiply defined labels here. Indeed, fig:sub2 appears twice. I'd rather not touch this, because I don't know where we refer to this label. Noleen?}

	\paragraph*{Clause selection gadget}
	For every $k\in [0,m]$ we introduce three vertices $p^k, q^k, r^k$ and add edges $p^kq^k$ and $p^kr^k$ of weight $1$ and $q^kr^k$ of weight $\rho$.
	
	We add edges $z_{n,b_n}p^0$, $p^{m}z_{1,1}$ of weight $5$. Additionally, for every $k\in [0,m-1]$ we add edges $p^{k}p^{k+1}$ of weight $5$. Hence all $z_{i,j}$'s and all $p^k$'s are contained in a cycle. 
	
	For every $k\in [0,m]$ we add edges $u_{n,b_n}^kq^k$, $q^ku_{1,1}^k$ of weight $5$. Additionally, for every $k\in [0,m]$ we add edges $v_{n,b_n}^kr^k$, $r^kv_{1,1}^k$ of weight $5$. Hence for each $k\in [0,m]$ the $u_{i,j}^k$'s together with $q^k$ form a cycle and the $v_{i,j}^k$'s together with $r^k$ form a cycle. For an illustration see \cref{fig:sub2}.
	
	\paragraph*{Auxiliary gadgets} 
	For every $i\in [n]$ we take one copy $(H_i,w_\rho)$ of $(H,w_\rho)$ and let  $h^{i}$ be the vertex $h$ in the copy $H_i$. We $9$-neighbourhood attach $(H_i,w_\rho)$ at $\{y_i,\lnot y_i\}$.

	Furthermore,  for every vertex $u\in \{\overline{s}_{i,j}^k,u_{i,j}^k,v_{i,j}^k: i\in [n], j\in [b_i],k\in [0,m]\}$ we introduce one copy $(H_u,w_\rho)$ of $(H,w_\rho)$ and let  $h_u$ be the vertex $h$ in the copy $H_u$. We $10$-neighborhood attach $(H_u,w_\rho)$ at $\{u\}$ for every $u\in \{\overline{s}_{i,j}^k: i\in [n], j\in [b_i],k\in [0,m]\}$. Additionally, if $a_{i,j}^k=0$ we $10$-neighborhood attach $(H_u,w_\rho)$ at $\{u\}$ for every $u\in \{u_{i,j}^k: i\in [n], j\in [b_i],k\in [0,m]\}$ and we $9$-neighborhood attach $(H_u,w_\rho)$ at $\{u\}$ for every $u\in \{v_{i,j}^k: i\in [n], j\in [b_i],k\in [0,m]\}$. Otherwise, if $a_{i,j}^k=1$ we $9$-neighborhood attach $(H_u,w_\rho)$ at $\{u\}$ for every $u\in \{u_{i,j}^k: i\in [n], j\in [b_i],k\in [0,m]\}$ and we $10$-neighborhood attach $(H_u,w_\rho)$ at $\{u\}$ for every $u\in \{v_{i,j}^k: i\in [n], j\in [b_i],k\in [0,m]\}$.
	
	Additionally, for every vertex $u\in \{p^k,q_k,r_k,z_{i,j}: i\in [n], j\in [b_i],k\in [0,m]\}$ we introduce one copy $(H_u,w_\rho)$ of $(H,w_\rho)$ and let  $h_u$ be the vertex $h$ in the copy $H_u$. We $10$-neighborhood attach $(H_u,w_\rho)$ at $\{u\}$.
	
	Finally, for every vertex $u\in \{s_{i,j}^k: i\in [n], j\in [b_i],k\in [0, m+1]\}$ we introduce one copy $(H_u,w_\rho)$ of $(H,w_\rho)$ and let  $h_u$ be the vertex $h$ in the copy $H_u$. We $9$-neighborhood attach $(H_u,w_\rho)$ at $\{u\}$ for $u=s_{i,j}^k$. For an illustration of where the auxiliary gadgets are attached see \cref{fig:test}.\\
	
	This concludes the construction of graph $(G,w)$. First observe that the degree of $G$ is at most $20$. To see this, observe that in $G\Big[\{s_{i,j}^k: j\in [b_i],k\in [0,m+1]\}\cup \{p_k,q_k,r_k,\overline{s}_{i,j}^k,t_{i,j}^k,u_{i,j}^k,v_{i,j}^k,y_i,\lnot y_i,z_{i,j}:i\in [n],j\in [b_i],k\in [0,m]\}\Big]$ every vertex has degree at most $4$ except $y_i,\lnot y_i$ who have degree at most $2$. Attaching a copy of $(H,w_\rho)$ at a single vertex increases the degree of that vertex by $6$ (attaching $(H,w_\rho)$ at two vertices increases the degree by $7$) and the degree of all its neighbours by $1$, yielding a degree bound of $14$ for vertices in $\{s_{i,j}^k: j\in [b_i],k\in [0,m+1]\}\cup \{p_k,q_k,\overline{s}_{i,j}^k,r_k,t_{i,j}^k,u_{i,j}^k,v_{i,j}^k,y_i,\lnot y_i,z_{i,j}:i\in [n],j\in [b_i],k\in [0,m]\}$. On the other hand, the degree of vertices within any copy of $(H,w_\rho)$ which was attach is at most $6$ plus the number of neighbours of the vertex at which the copy of $(H,w_\rho)$ was attached. %vertices in $G\Big[\{s_{i,j}^k: j\in [b_i],k\in [0,m+1]\}\cup \{p_k,q_k,\overline{s}_{i,j}^k,r_k,t_{i,j}^k,u_{i,j}^k,v_{i,j}^k,y_i,\lnot y_i,z_{i,j}:i\in [n],j\in [b_i],k\in [0,m]\}\Big]$.
	Also observe, that $\rho<- \max_{u\in V(G)}\sum_{uv\in E(G),\atop{w(uv)>0}}w(uv)$ and hence we can use Lemma~\ref{lem:auxiliaryNeighborhoodGadget}.
	
	We first argue that $G$ has path-width in $\mathcal{O}(\log n)$. To this end, define 
	$$B_0=\{p^k,q^k,r^k:k\in [0,m]\}\cup \bigcup_{u\in \{p^k,q^k,r^k:k\in [0,m]\}}V(H_u) $$ 
	and for every $i\in [n]$  define 
	\begin{align*}
		B_i=&\{s_{i,j}^k: j\in [b_i],k\in [0,m+1]\}\cup \{\overline{s}_{i,j}^k,t_{i,j}^k,u_{i,j}^k,v_{i,j}^k, y_i,\lnot y_i, z_{i,j}: j\in [b_i],k\in [0,m]\}\\&\cup V(H_i)\cup \bigcup_{u\in \{s_{i,j}^k: j\in [b_i],k\in [0,m+1]\}  \atop{\cup\{\overline{s}_{i,j}^k, u_{i,j}^k,v_{i,j}^k, z_{i,j}: j\in [b_i],k\in [0,m]\}}}V(H_u).
	\end{align*}
	We observe that $B_0$ has size $21+21 m$ while $B_i$ has size at most $125+87 m$. By construction there is no edges between $B_i$ and $B_j$ if $|i-j|>1$. Hence we get a path-decomposition of $G$ of width at most $271+195 m$ by placing all vertices of $B_0,B_i,B_{i+1}$ in the bag of the $i$th vertex of the path. We now argue that $(G,w)$ is core stable if and only if $\phi$ is satisfiable.\\
	
	First assume that $(G,w)$ is core stable and let $\mathcal{P}$ be a core stable partition of $V(G)$. Without loss of generality we can assume that $\mathcal{P}$ is connected. 
	By Lemma~\ref{lem:auxiliaryNeighborhoodGadget} property~\ref{prop:gadgetP1Alternative}, we know that  $\{u,h_u\}\in\mathcal{P}$ for every $u\in \{s_{i,j}^k: j\in [b_i],k\in [0,m+1]\}\cup \{p^k,q^k,r^k,\overline{s}_{i,j}^k,u_{i,j}^k,v_{i,j}^k,z_{i,j}:i\in [n],j\in [b_i],k\in [0, m]\}$ and   either  $\{h^i,y_i\}\in \mathcal{P}$ or $\{h^i\lnot y_i\}\in \mathcal{P}$  for every $i\in [n]$. Additionally, by Lemma~\ref{lem:auxiliaryNeighborhoodGadget} property~\ref{prop:gadgetP2} we know that  
	$$P\subseteq \bigcup_{i\in [n]}V(H_i)\setminus\{h^i\}\cup \bigcup_{u\in \{s_{i,j}^k: j\in [b_i],k\in [0,m+1]\}\atop{\cup\{p^k,q^k,r^k,t_{i,j}^k,u_{i,j}^k,v_{i,j}^k,z_{i,j}:i\in [n],j\in [b_i],k\in [0,m]\}}}V(H_u)\setminus \{h_u\}$$ 
	if $P\in \mathcal{P}$ contains any vertex from this union of auxiliary gadgets. As  $G[\{t_{i,j}^k:i\in [n],j\in [b_i],k\in [0, m]\}]$ contains only singletons and we assumed that $\mathcal{P}$ is connected, we further know that $\{t_{i,j}^k\}\in \mathcal{P}$ for every $i\in [n]$, $j\in [b_i]$, $k\in [0,m]$.  
	
	We define an assignment $\alpha:\{x_1,\dots,x_n\}\rightarrow \{0,1\}$ in the following way. If $\{h^i,y_i\}\in \mathcal{P}$, then we set $\alpha(x_i)=1$ and if $\{h^i\lnot y_i\}\in \mathcal{P}$, then we set $\alpha(x_i)=0$. We claim that $\alpha$ is a satisfying assignment for $\phi$. Towards a contradiction assume that this is not the case and assume that $c_\ell$ is a clause which does not evaluate to $1$ under the assignment $\alpha$. Let $\ell=\sum_{k=0}^{m}a_k 2^k$ be the binary representation of $\ell$. Define a set $X$ as follows. 
	\begin{align*}
		X=&\Big\{p^k:k\in [0,m]\Big\} 
		\cup \Big\{q^k:k\in [0,m], a_k=0\Big\}\cup \Big\{r^k:k\in [0,m], a_k=1\Big\}\\
		\cup &\Big\{u_{i,j}^k:i\in [n],j\in [b_i],k\in [0,m],a_k=0\Big\}\\
		\cup &\Big\{v_{i,j}^k:i\in [n],j\in [b_i],k\in [0,m],a_k=1\Big\}\\
		\cup &\Big\{t_{i,j}^k:i\in [n],j\in [b_i], k\in [0,m], a_{i,j}^k=a_k\Big\}\\
		\cup & \Big\{s_{i,j}^k:i\in [n],j\in [b_i], k\in [0,m+1],k\geq \min\big(\{k'\in [0,m]:a_{i,j}^{k'}\not= a_{k'}\}\cup \{m+2\}\big)\Big\}\\
		\cup & \Big\{\overline{s}_{i,j}^k:i\in [n],j\in [b_i], k\in [0,m],k\geq \min\big(\{k'\in [0,m]:a_{i,j}^{k'}\not= a_{k'}\}\cup \{m+2\}\big)\Big\}\\
		\cup &\Big\{y_i:i\in [n], \{y_i,h^i\}\notin \mathcal{P}\Big\}\cup \Big\{\lnot y_i:i\in [n], \{\lnot y_i,h^i\}\notin \mathcal{P}\Big\}\cup \Big\{z_{i,j}:i\in [n],j\in [b_i]\Big\}.
	\end{align*}
	In the following we argue that $X$ is a blocking coalition of $\mathcal{P}$ which yields a contradiction. First observe that no two vertices in $X$ are connected by an edge of weight $\rho$. Furthermore, all edges in $G[X]$ have positive weight, apart from edges $s_{i,j}^kt_{i,j}^k$. Hence, it is sufficient to only consider positive weight edges and edges of the form $s_{i,j}^kt_{i,j}^k$.
	Let $i\in [n]$, $j\in [b_i]$, $k\in [0,m]$ be arbitrary but fixed. 
	
	First note that $p^k$ has utility $10$ in $\mathcal{P}$ as $\{p^k,h_{p^k}\}\in \mathcal{P}$. As both $p^{k-1}\in X$ ($z_{n,b_n}\in X$ if $k=0$) and $p^{k+1}\in X$ ($z^{1,1}\in X$ if $k=m$) and additionally either $q^k\in X$ or $r^k\in X$, we conclude that $p^k$ has improved its utility from $10$ to $11$ by joining $X$. 
	
	Next observe that both $q^k$ and $r^k$ have utility $10$ in $\mathcal{P}$. If $a_k=0$ then $q^k\in X$ ($r^k\notin X$) and so are $p^k$ as well as both $u_{n,b_n}^k$ and $u_{1,1}^k$. Hence, $q^k$ has utility $11>10$ in $X$. The same argument shows that $r^k$ has utility $11>10$ in $X$ in the case that $a_k=1$. 
	
	Furthermore,  $u_{i,j}^k$ has utility $10$ if $a_{i,j}^k=0$ and $9$ otherwise. Assume that $a_k=0$ and hence $u_{i,j}^k\in X$ (and $v_{i,j}^k\notin X$).  In the case that $a_{i,j}^k=0$, by construction $t_{i,j}^k\in X$ and $u_{i,j}^k t_{i,j}^k\in E(G)$ and hence $u_{i,j}^k$ has utility $11$ in $X$. On the other hand, if $a_{i,j}^k=1$, then $u_{i,j}^k$ has utility $10$ in $X$. In both cases $u_{i,j}^k$ improves its utility by joining $X$. An analogous argument shows that, if $a_k=1$ and hence $v_{i,j}^k\in X$, then $v_{i,j}^k$ improves its utility by joining $X$. 
	
	As $\{t_{i,j}^k\}\in \mathcal{P}$, $t_{i,j}^k$ has utility $0$ in $\mathcal{P}$. In case $a_{i,j}^k=a_k$ and $t_{i,j}^k\in X$ we need to consider two cases. If $a_{i,j}^k=a_k=0$ then $t_{i,j}^k$ is adjacent to $u_{i,j}^k$, $w(t_{i,j}^ku_{i,j}^k)=2$ and $u_{i,j}\in X$. Even if $t_{i,j}^k$'s other neighbour $s_{i,j}^k\in X$, $t_{i,j}^k$ still has utility $1$ in $X$. Equivalently, if $a_{i,j}^k=a_k=1$, then $t_{i,j}^k$ has utility at least $1$ in $X$. 
	
	Consider $\overline{s}_{i,j}^k$ and assume that $\overline{s}_{i,j}^k\in X$. By definition, this implies that $k\leq \min\big(\{k'\in [0,m]:a_{i,j}^{k'}\not= a_{k'}\}\cup \{m+2\}\big)$ and hence both $s_{i,j}^k\in X$ and $s_{i,j}^{k+1}\in X$. Hence, $\overline{s}_{i,j}^k$ has utility $19$ in $X$ while $\overline{s}_{i,j}^k$ has utility $10$ in $\mathcal{P}$.
	
	Next consider $s_{i,j}^{m+1}$ and assume that $s_{i,j}^{m+1}\in X$. Hence, $m+1\geq \min\big(\{k'\in [0,m]:a_{i,j}^{k'}\not= a_{k'}\}\cup \{m+2\}\big)$. As $m+1$ cannot be equal to $\min\big(\{k'\in [0,m]:a_{i,j}^{k'}\not= a_{k'}\}\cup \{m+2\}\big)$ we know that $\overline{s}_{i,j}^m\in X$. Hence the utility of $s_{i,j}^{m+1}$ in $X$ is $10$ (as $z_{i,j}$ is also in $X$) while the utility of $s_{i,j}^{m+1}$ is $9$ in $\mathcal{P}$. Next consider $s_{i,j}^k$, $k\leq m$
	and assume that $s_{i,j}^k\in X$ and hence
	$k\geq \min\{k'\in [0, m]:a_{i,j}^{k'}\not= a_{k'}\}$. First consider the case that $k= \min\{k'\in [0,m]:a_{i,j}^{k'}\not= a_{k'}\}$. As $a_{i,j}^k\not=a_k$, we know that $t_{i,j}^k\notin X$. As by construction $\overline{s}_{i,j}^{k}\in X$ and $w(s_{i,j}^k\overline{s}_{i,j}^{k})=10$ we know that $s_{i,j}^k$ has utility $10$ in $X$ while its utility is $9$ in $\mathcal{P}$. On the other hand, if $k> \min\{k'\in [0,m]:a_{i,j}^{k'}\not= a_{k'}\}$ then both $\overline{s}_{i,j}^{k-1}\in X$ and $\overline{s}_{i,j}^{k}\in X$ and hence $s_{i,j}^k$ has utility at least $18$ (even if $t_{i,j}^k\in X$). 
	
	Now consider $y_i$ and observe that if $y_i\in X$, then $y_i$ has utility $0$ in $\mathcal{P}$. As $y_i$ has at least one neighbor in the set $\{z_{i,j}:i\in [n],j\in [b_i]\}$ (by our assumption that both literals $x_i, \lnot x_i$ appear in $\phi$) and the corresponding edge has positive weight, $y_i$ improves its utility by joining $X$. Equivalently, for $\lnot y_i$. 
	
	Finally, consider the vertex $z_{i,j}$ and observe that its utility in $\mathcal{P}$ is $10$. First assume that $s_{i,j}^{m+1}\in X$. In this case $z_{i,j}$ has utility $11$ in $X$ as $\{p^1,p^m\} \cup \{z_{i,j}:i\in [n],j\in [b_i]\}\subseteq X$. On the other hand, if $s_{i,j}^{m+1}\notin X$ then $\min\big(\{k'\in [0,m]:a_{i,j}^{k'}\not= a_{k'}\}\cup \{m+2\}\big)=m+2$ by construction of $X$. Hence, $a_{i,j}^{k'}=a_{k'}$ for every $k'\in [0,m]$ which implies that $\sum_{k'=0}^m a_{i,j}^{k'}2^{k'}=\ell$. By construction this means that  $x_i$ appears in $c_\ell$. Assume that the literal $x_i$ is contained in $c_\ell$ (the case that the literal $\lnot x_i$ is contained in $c_\ell$ is analogous). As $x_i$ appears in $c_\ell$ and $c_\ell$ evaluates to $0$ under the assignment $\alpha$, we know that $\alpha(x_i)=0$. By definition of $\alpha$ this means that $\{h^i,\lnot y_i\}\in \mathcal{P}$ and hence $y_i\in X$. Therefore, $z_{i,j}$ has utility $11$ in $X$. As $z_{i,j}$ improves its utility in both cases we have argued that $X$ is a blocking coalition contradicting our assumption that $\mathcal{P}$ is core stable. Hence, $\alpha$ is a satisfying assignment for $\phi$. \\

	On the other hand, assume that $\phi$ is satisfiable and let $\alpha:\{x_1,\dots,x_n\}\rightarrow \{0,1\}$ be a satisfying assignment of $\phi$. We define a partition $\mathcal{P}$ of $V(G)$ as follows. For $i\in [n]$ we include sets $\{h^i,y_i\}$ as well as $\{\lnot y_i\}$ in $\mathcal{P}$ if $\alpha(x_i)=1$ and we include sets $\{h^i, \lnot y_i\}$ as well as $\{y_i\}$ in $\mathcal{P}$ if $\alpha(x_i)=0$. We further include the core stable partition of $H_i\setminus \{h^i\}$ defined in Section~\ref{sec:auxiliaryGadget} in $\mathcal{P}$.
	For each vertex $u\in \{s_{i,j}^k: j\in [b_i],k\in [0,m+1]\}\cup \{p_k,q_k,\overline{s}_{i,j}^k,r_k,u_{i,j}^k,v_{i,j}^k,z_{i,j}:i\in [n],j\in [b_i],k\in [0,m]\}$ we include the set $\{u,h_u\}$ in $\mathcal{P}$ and the core stable partition $\mathcal{P}_{(H_u,w_{\rho})}$ of $H_u\setminus \{h_u\}$ defined in Section~\ref{sec:auxiliaryGadget}. Finally, for every $u\in \{t_{i,j}^k:i\in [n],j\in [b_i],k\in [0, m]\}$ we include the singleton $\{u\}$ in $\mathcal{P}$. Towards a contradiction, assume that $\mathcal{P}$ is not a core stable partition and $X$ is a blocking coalition. In the following, we show that $X$ must have a particular structure, i.e. $X$ must be of the form as the blocking coalition in the previous direction of the proof. The property of the coalition $X$ being blocking can then be translated into an argument that $\alpha$ is not a satisfying assignment yielding the desired contradiction.
	
	We define $Z=\{p^k,z_{i,j}:i\in [n], j\in [b_i],k\in [0, m]\}\subseteq X$ and for every $k\in [0,m]$ we define $U^k=\{q^k,u_{i,j}^k:i\in [n],j\in [b_i]\}$, $V^k=\{r^k,v_{i,j}^k:i\in [n],j\in [b_i]\}$, $T^k_0=\{t_{i,j}^k:i\in [n], j\in [b_i],a_{i,j}^k=0\}$ and $T_1^k=\{t_{i,j}^k:i\in [n], j\in [b_i],a_{i,j}^k=1\}$. The following three claims describe precisely  which compositions of vertices are valid for $X$.
	\begin{claim}\label{claim:blockingCoalitionContainsCycles}
		For every $k\in [0,m]$ the following three statements hold.
		\begin{enumerate}[left=6pt , label=$(\roman*)$]
			\item If any $z\in Z$ is contained in $X$, then $Z\subseteq X$ and  either $U^k\cup T_0^k\subseteq X$ or $V^k\cup T_1^k\subseteq X$.
			\item If any $u\in U^k\cup T_0^k$ is contained in $X$,  then $U^k\cup T_0^k\subseteq X$ and $Z\subseteq X$.
			\item If any $v\in V^k\cup T_1^k$ is contained in $X$,  then $V^k\cup T_1^k\subseteq X$ and $Z\subseteq X$.
		\end{enumerate}
	\end{claim} 
	\begin{proof}
		First note that the vertices of $Z$ form a cycle $C_Z$ and for every $k\in [0,m]$ the vertices of $U^k$ form a cycle $C_U^k$ and the vertices of $V^k$ form a cycle $C_V^k$.  Since every $u\in \bigcup_{k\in [0,m]}U^k\cup \bigcup_{k\in [0,m]}V^k$ has utility at least $9$ in $\mathcal{P}$, if $u\in X$ then its utility has to be at least $10$ in $X$. Fix some $u\in \bigcup_{k\in [0,m]}U^k\cup \bigcup_{k\in [0,m]}V^k$ and assume $u\in X$. By Lemma~\ref{lem:auxiliaryNeighborhoodGadget} property \ref{prop:gadgetP3}, we know that $h_u\notin X$. The only remaining positive weight edges incident to $u$ are the two edges on the respective cycle $C_Z, C_U^k,C_V^k$ of weight $5$ and at most one additional edge of weight at most $2$. Hence $u$ can only have utility $10$ in $X$ if both neighbors of $u$ on the respective cycle $C_Z, C_U^k,C_V^k$ are contained in $X$. Inductively, this proves that if any $z\in Z$ is contained in $X$, then $Z\subseteq X$; if $u\in U^k$ is contained in $X$ for some $k\in [0,m]$, then $U^k\subseteq X$; and if $v\in V^k$ is contained in $X$ for some $k\in [0,m]$, then $V^k\subseteq X$.
		
		Towards showing (ii), we first observe the following. If $t_{i,j}^k\in X$ for some $i\in [n]$, $j\in [b_i]$, $k\in [0,m]$ and $a_{i,j}^k=0$, then $u_{i,j}^k$ is the only neighbor of $t_{i,j}^k$ with positive edge weight and hence $u_{i,j}^k\in X$. But if $u_{i,j}^k\in X$ then $U^k\subseteq X$.  We only have left to argue that if $U^k\subseteq X$ for some $k\in [0,m]$, then $T_0^k\subseteq X$. Hence, assume that $U^k\subseteq X$. By construction, the utility of  $u_{i,j}^k$ in $\mathcal{P}$ is $10$ for every $i\in [n]$, $j\in [b_i]$ for which $a_{i,j}=0$. Since $u_{i,j}^k$ gets utility $10$ by $U^k\subseteq X$ we know that $t_{i,j}^k\in X$ if $a_{i,j}^k=0$. Hence $T_0^k\subseteq X$. 
		We can prove (iii) with an analogous argument.
		
		Towards showing (i), observe that we only have left to argue that if $Z\subseteq X$, then either $U^k\cup T_0^k\subseteq X$ or $V^k\cup T_1^k\subseteq X$ for every $k\in [0,m]$. Observe that $p^k$ has utility $10$ in $\mathcal{P}$ for every $k\in [0,m]$. As $p^k$ gets utility at most $10$ by $Z\subseteq X$ we know that either $q^k\in X$ or $r^k\in X$. But then by our previous observation either $U^k\subseteq X$ or $V^k\subseteq X$. But then $U^k\cup T_0^k\subseteq X$ or $V^k\cup T_1^k\subseteq X$ by (ii) and (iii).
	\end{proof}
	\begin{claim}\label{claim:notAllInGadget}
		For every $i\in [n]$, $j\in [b_i]$, $k\in [0,m+1]$ the following two statements hold
		\begin{enumerate}[left=6pt , label=$(\roman*)$]
			\item If $k\not= m+1$,   $s_{i,j}^k\in X$ if and only if  $\overline{s}_{i,j}^k\in X$.
			\item If $s_{i,j}^k\in X$, then $s_{i,j}^{k'}\in X$ for every $k'>k$, $k'\leq m+1$ and $Z\subseteq X$.
			\item If $s_{i,j}^k\in X$, then either $k\leq m$ and $t_{i,j}^k\notin X$ or $k>0$ and $s_{i,j}^{k-1}\in X$.
		\end{enumerate}     
	\end{claim}
	\begin{proof}
		To show (i), assume $s_{i,j}^k\in X$. Observe that $s_{i,j}^k\in X$ has utility $9$ in $\mathcal{P}$. Additionally note that $h_{s_{i,j}^k}\notin X$ by \cref{lem:auxiliaryNeighborhoodGadget} property \ref{prop:gadgetP3} and $s_{i,j}^k$ is incident to at most two additional edges of positive weight (in case $k=0$ it is only one edge) of which only one, the edge $s_{i,j}^k\overline{s}_{i,j}^k$, has weight larger than $9$. Therefore, $\overline{s}_{i,j}^k$ must be contained in $X$. On the other hand, if $\overline{s}_{i,j}^k\in X$, then $s_{i,j}^k\in X$. To see this note that $\overline{s}_{i,j}^k$ has utility $10$ in $\mathcal{P}$ and (excluding the edge $\overline{s}_{i,j}^kh_{\overline{s}_{i,j}^k}$) is incident to two edges, one of weight $9$ and one of weight $10$. As $h_{\overline{s}_{i,j}^k}\notin X$, this implies that both $s_{i,j}^k$ and $s_{i,j}^{k+1}$ have to be included in $X$.
		
		Towards showing (ii), assume that $s_{i,j}^k\in X$. First assume that $k\leq m$ and hence $\overline{s}_{i,j}^k\in X$ by (i).   As argued previously, this implies that $s_{i,j}^{k+1}\in X$. Inductively, this proves that $s_{i,j}^{k'}\in X$ for every $k'>k$, $k'\leq m+1$. %Observe that $s_{i,j}^k$ has utility $9$ in $\mathcal{P}$. Furthermore, $s_{i,j}^k$ is incident to only two edges of positive weight and one of them is of weight $9$ and therefore, the other neighbor of $s_{i,j}^k$ has to be in $X$.   
		Additionally,  $s_{i,j}^{m+1}\in X$ implies that $z_{i,j}\in X$. To see this, observe that the edges of positive weight incident to $s_{i,j}^{m+1}$ are $s_{i,j}^{m+1}h_{s_{i,j}^{m+1}}$, $s_{i,j}^{m+1}\overline{s}_{i,j}^{m}$ of weight $9$ and $s_{i,j}^{m+1}z_{i,j}$ of weight $1$. As $h_{s_{i,j}^{m+1}}\notin X$ by \ref{lem:auxiliaryNeighborhoodGadget} property \ref{prop:gadgetP3} we get $z_{i,j}\in X$. Using Claim~\ref{claim:blockingCoalitionContainsCycles}~(i), this proves (i).
		
		To argue (iii) holds, observe that if $s_{i,j}^0\in X$ and $t_{i,j}^0\in X$,  then the utility of $s_{i,j}^0$  in $X$ is at most $9$ as $w(s_{i,j}^0t_{i,j}^0)=-1$ and $w(s_{i,j}^0s_{i,j}^1)=10$. This contradicts that the utility of $s_{i,j}^0$ is $9$ in $\mathcal{P}$. Now assume $k>0$, $k\leq m$ and $s_{i,j}^k\in X$. 
		If both $t_{i,j}^k\in X$ and $s_{i,j}^{k-1}\notin X$, then the utility of $s_{i,j}^k$ in $X$ is at most $9$ as $s_{i,j}^{k-1}\notin X$ implies $\overline{s}_{i,j}^{k-1}\notin X$ by (i). As the utility of $s_{i,j}^k$ is $9$ in $\mathcal{P}$ this yields a contradiction. On the other hand, if $k=m+1$, $s_{i,j}^{m+1}\in X$ and $s_{i,j}^{m}\notin X$, then $s_{i,j}^{m+1}$ has utility $1$ in $X$ as $s_{i,j}^{m}\notin X$ implies $\overline{s}_{i,j}^{m}\notin X$ by (i). This yields a contradiction as $s_{i,j}^{m+1}$ has utility $10$ in $\mathcal{P}$.
		Therefore, we have shown that (ii) is true. 
	\end{proof}
	\begin{claim}\label{claim:satisfiedLiterals}
		For every $i\in [n]$, $j\in [b_i]$ the following three statements hold
		\begin{enumerate}[left=6pt , label=$(\roman*)$]
			\item If either $y_i\in X$ or $\lnot y_i\in X$, then $Z\subseteq X$.
			\item If $z_{i,j}\in X$, then either $s_{i,j}^{m+1}\in X$ or either $y_i\in X$ if the clause of $\phi$ containing the $j$th appearance of $x_i$ contains $x_i$ or $\lnot y_i\in X$ if the clause of $\phi$ containing the $j$th appearance of $x_i$ contains $\lnot x_i$.
			\item If $\alpha(x_i)=1$, then $y_i$ cannot be in $X$ and if $\alpha(x_i)=0$ then $\lnot y_i$ cannot be in $X$.
		\end{enumerate}     
	\end{claim}
	\begin{proof}
		To prove (i), recall that we assumed that $\phi$ contains both the literal $x_i$ as well as $\lnot x_i$ and hence both $y_i$ and $\lnot y_i$ must have a neighbor in $Z$ by construction. By Lemma~\ref{lem:auxiliaryNeighborhoodGadget} property~\ref{prop:gadgetP3}, we know that $h^i\notin X$. Since every other edge incident to either $y_i$ and $\lnot y_i$ of positive weight is also incident to some vertex in $Z$, we conclude that some $z\in Z$ must be contained in $X$. Hence, by Claim~\ref{claim:blockingCoalitionContainsCycles}~(i) $Z\subseteq X$.
		
		Towards showing (ii), first observe that $z_{i,j}$ has utility $10$ in $\mathcal{P}$. If $z_{i,j}\in X$, then $z_{i,j}$ needs to have utility at least $11$ in $X$. Assume that the clause of $\phi$ containing the $j$th appearance of $x_i$ contains $x_i$ (the case that the clause of $\phi$ containing the $j$th appearance of $x_i$ contains $\lnot x_i$ can be argued analogously). If both $s_{i,j}^{m+1}\notin X$ and $y_i\notin X$, then $z_{i,j}$ can have utility at most $10$ in $X$ which is not sufficient. Hence,  (ii) holds.

		Lastly, to argue (iii),  observe that if $\{y_i,h^i\}\in \mathcal{P}$, then $y_i$ has utility $9$ in $\mathcal{P}$ and can therefore not join $X$.  As $\{y_i,h^i\}\in \mathcal{P}$ if $\alpha(x_i)=1$, we obtain the first part of the statement. An analogous argument shows that $\lnot y_i$ cannot be in $X$ if $\alpha(x_i)=0$. 
	\end{proof}
	By Lemma~\ref{lem:auxiliaryNeighborhoodGadget} property~\ref{prop:gadgetP2} we know that $X\subseteq \{s_{i,j}^k: j\in [b_i],k\in [0,m+1]\}\cup \{p_k,q_k,r_k,\overline{s}_{i,j}^k,t_{i,j}^k,u_{i,j}^k,v_{i,j}^k,z_{i,j}:i\in [n],j\in [b_i],k\in [0,m]\}$. 
	Since $X$ cannot be empty, combining Claim~\ref{claim:blockingCoalitionContainsCycles}, Claim~\ref{claim:notAllInGadget}~(i) and Claim~\ref{claim:satisfiedLiterals}~(i) we obtain  that $Z\subseteq X$. Consequently, by Claim~\ref{claim:blockingCoalitionContainsCycles}~(i)  either $U^k\subseteq X$ or $V^k\subseteq X$ for every $k\in [0,m]$. Note that they cannot both be contained in $X$ as $w(q^kr^k)=\rho$. For every $k\in [0,m]$ we pick $a_k\in \{0,1\}$ as follows and define $\ell=\sum_{k=0}^{m}a_k 2^k$. We let $a_k$ be $0$ if $U^k\subseteq Z$ and we let $a_k$ be $1$ if $V^k\subseteq Z$. In the following, we argue that no literal contained in the clause $c_\ell$ is satisfied. To prove this we use the following claim. Recall that for fixed $i\in [n]$ and $j\in [b_i]$ and we defined $\sum_{k=0}^{m}a_{i,j}^k 2^k$ where $a_{i,j}^k\in \{0,1\}$ be the binary representation of $\ell$ where $c_\ell$ is the clause containing the $j$th appearance of $x_i$.
	\begin{claim}\label{claim:encodingClause}
		If  $\ell=\sum_{k=0}^{m}a_{i,j}^k 2^k$ for some $i\in [n]$, $j\in [b_i]$, then $s_{i,j}^{m+1}\notin X$.
	\end{claim}
	\begin{proof}
		Towards a contradiction assume that $s_{i,j}^{m+1}\in X$. 
		Let $k\in [0,m+1]$ be the minimum index such that $s_{i,j}^k\in X$, i.e., either $k=0$ or $s_{i,j}^{k'}\notin X$ for every $k'<k$. Observe, that $k\leq m$ as if $s_{i,j}^{m+1}\in X$, then $s_{i,j}^m$ must also be in $X$ by \cref{claim:notAllInGadget}~(iii). By Claim~\ref{claim:notAllInGadget}~(iii), $k$ being the minimum index for which $s_{i,j}^k\in X$  means that $t_{i,j}^k\notin X$.
		Assume that $a_{i,j}^k=0$ (the case that $a_{i,j}^k=1$ works analogously) which implies that $t_{i,j}^k\in U^k$. Therefore, $U^k\not\subseteq X$ which implies that $a_k=1$ by definition. Since $a_{i,j}^k\not= a_k$ this contradicts that 
		$\ell=\sum_{k=0}^{m}a_{i,j}^k 2^k$. 
	\end{proof}
	Now assume that for some fixed $i\in [n]$, $j\in [b_i]$ the clause $c_\ell$ contains the $j$th appearance of variable $x_i$ and hence $\sum_{k=0}^{m}a_{i,j}^k 2^k=\ell$. We assume that $c_\ell$ contains the literal $x_i$ (the case that $c_\ell$ contains the literal $\lnot x_i$ works analogously). As $z_{i,j}\in X$ and $s_{i,j}^{m+1}\notin X$ by Claim~\ref{claim:encodingClause}, we get that $y_i$ must be contained in $X$ by Claim~\ref{claim:satisfiedLiterals}~(ii). Since $y_i\in X$ we get that $\alpha(x_i)=0$ and hence the literal $x_i$ contained in $c_\ell$ is not satisfied. Since no literal in $c_\ell$ can be satisfied, we obtain that $c_\ell$ is not satisfied contradicting the assumption that $\alpha$ was a satisfying assignment.

\end{proof}
\end{toappendix}

\section{$k$-Core Stability}

In this section we consider the complexity of finding and verifying $k$-core
stable partitions, when the size of the allowed blocking coalitions $k$ is a
parameter. Even though the two problems do become easier when $k$ is a fixed
constant (because we can check all possible blocking coalitions in polynomial
time), we show that it is likely that not much more can be gained from this
assumption: $k$-\textsc{CSV} is coW[1]-hard parameterized by $k$
(\cref{thm:csv:W-hard:k}), while $k$-\textsc{CS} is NP-complete even if $k\ge
3$ is a fixed constant (\cref{thm:3csf:3col}). On the positive side, we do show
that finding $2$-core stable partitions is in P, but it is worth noting that
the fact that we consider undirected graphs is crucial to obtain even this
small tractable case.

\begin{theorem}\label{thm:csv:W-hard:k} \ifthenelse{\boolean{short}}{\textup{($\star$)}}{}
	\textsc{$k$-Core Stability Verification}  is in XP when parameterized by $k$  whereas coW[1]-hard even on unweighted  graphs.
\end{theorem}
\begin{appendixproof}[Proof\ifthenelse{\boolean{short}}{of \cref{thm:csv:W-hard:k}}{}]
	The upper bound can be easily shown by brute force. That is, given a coalition structure $\PP$, for each coalition $X$ of size at most $k$, we check if each agent $v$ in $X$ has higher utility than in $\PP$. 
	The running time of brute force is $n^{O(k)}$.
	
	Then we show that \textsc{$k$-CSV} is W[1]-hard even on unweighted  graphs. We give a reduction from \textsc{$k$-Clique}.
	Given a graph $G$, we attach $k-2$ pendant vertices for each vertex in $V(G)$. Let $P_v$ be the set of pendant vertices for $v$.
	We set $\PP = \{P_v\cup \{v\}: v\in V(G)\}$ as a coalition structure to verify.
	
	In the following, we show that there exists a $k$-clique in $G$ if and only if there exists a blocking coalition for $\PP$.
	Let $C$ be a clique of size $k$ in $G$. For $C$, each vertex in $C$ has the utility $k-1$. Since $v\in V$ has the utility $k-2$ in $\PP$, $C$ is a blocking coalition for $\PP$.
	Conversely, let $X$ be a blocking coalition for $\PP$.
	Since vertices in $\bigcup_{v\in V(G)} P_v$ have maximum utility 1 in $\PP$, they do not join $X$.
	Thus, $X$ is a subset of $V(G)$.
	Since the utility of $v\in V(G)$ is $k-2$ in $\PP$ and $|X|=k$,  $X$ is a clique of size $k$.
\end{appendixproof}

\begin{theorem}\label{thm:csv:NP:k} \ifthenelse{\boolean{short}}{\textup{($\star$)}}{}
	Every graph admits a $2$-core stable partition and 
	\textsc{$2$-Core Stability}  can be solved in polynomial time.
\end{theorem}
\begin{appendixproof}[Proof\ifthenelse{\boolean{short}}{ of \cref{thm:csv:NP:k}}{}]
	Given
	a weighted graph $(G,w)$, start with the partition where every vertex is a singleton.
	Order the positive-weight edges in non-increasing order
	$e_1,e_2,\ldots,e_m$. For each $e_i$, do the following: if the endpoint of
	$e_i$ are currently singletons, merge them into a cluster of size 2;
	otherwise move to the next edge. The resulting partition $\mathcal{P}$ is $2$-core stable because if
	there was a blocking coalition of size 2, it would have to induce an edge
	$e_i=uv$. However,  when  $e_i$ is considered, at least one of $u,v$ was
	not a singleton. Therefore, the utility of that vertex  must  be  larger in $\mathcal{P}$ than in the
	coalition $\{u,v\}$ contradicting the assumption that $\{u,v\}$ is a blocking coalition. 
\end{appendixproof}

\begin{theorem}\label{thm:3csf:3col} \textup{($\star$)}
	For any fixed $k\ge 3$, \textsc{$k$-Core Stability}  is NP-complete on bounded degree graphs even if the weights are constant.
\end{theorem}

\begin{appendixproof}
	\begin{figure}
		\centering
		\centerline{\includegraphics[scale=0.9]{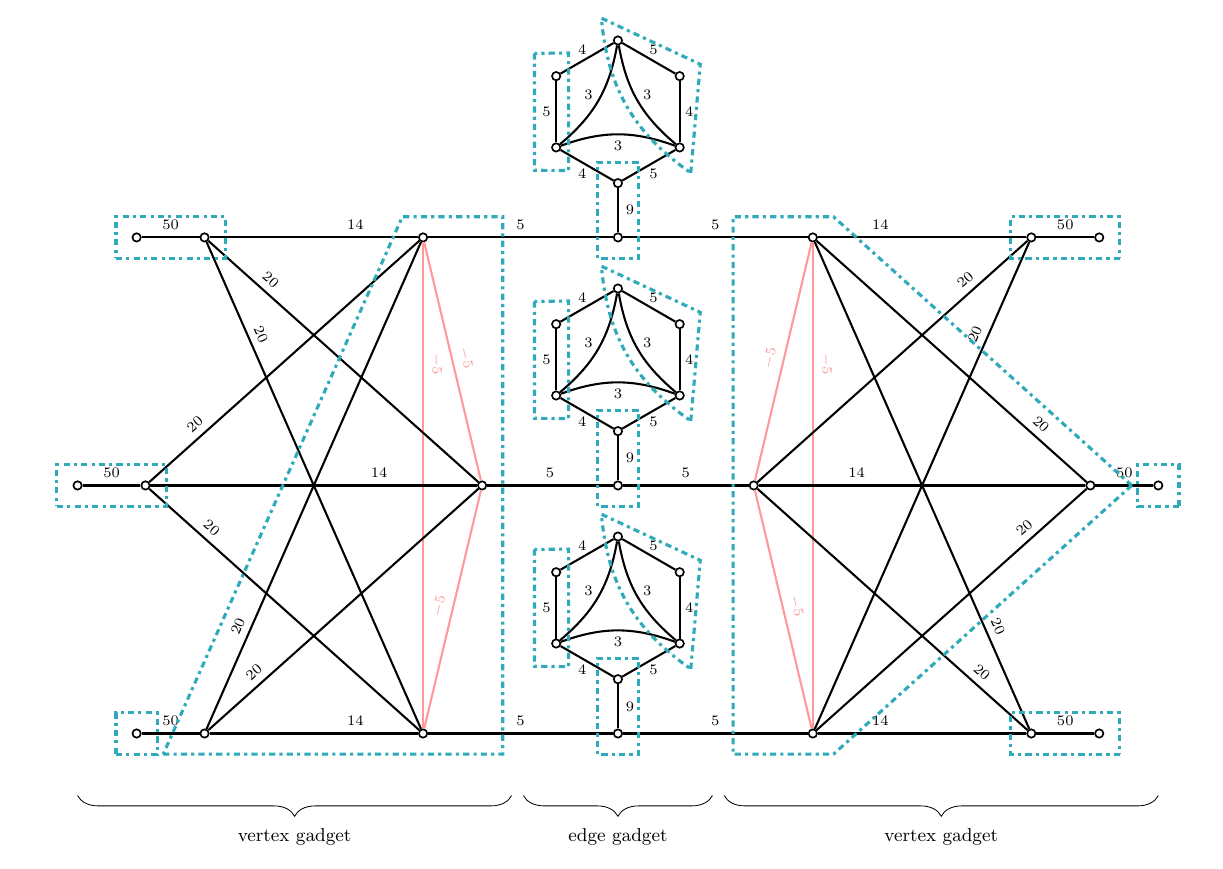}}
		\caption{The graph $\widehat{G}$ constructed in the proof of Theorem~\ref{thm:3csf:3col} for an instance of $3$-colouring $G$ consisting of a single edge. In the illustration edges of weight $\rho$ are omitted. The core stable partition indicated by dot-dashed boxes corresponds to coloring the two vertices of $G$ with two different colors.}
		\label{fig:3csf:3col}
	\end{figure}
	First observe that  $k$-\textsc{CS} is in NP by Theorem~\ref{thm:csv:W-hard:k}.
	
	We give a reduction from bounded degree 3-\textsc{Coloring}. Given a sub-cubic graph $G$ we construct an instance $(\widehat{G},\widehat{w})$ of $3$-\textsc{CS} as follows.  We let $\rho=-104$.
	\paragraph*{Vertex gadget}
	For every vertex $x\in V(G)$ we introduce vertices $u_1^x,\dots,u_9^x$ and the following edges. First we add edges $u_1^xu_2^x,u_2^xu_3^x,u_3^xu_1^x$ of weight $-5$, $u_i^xu_{i+3}^x$ of weight $14$ for every $i\in \{1,2,3\}$, $u_4^xu_8^x$, $u_4^xu_9^x$, $u_5^xu_7^x$, $u_5^xu_9^x$, $u_6^xu_7^x$, $u_6^xu_8^x$ of weight $20$ and finally $u_iu_{i+3}$ of weight $50$ for every $i\in \{4,5,6\}$. For every pair $i\not=j\in [9]$ for which there is no edge yet, we add an edge $u_i^xu_j^x$ of weight $\rho$. 
	\paragraph*{Edge gadget}
	%We further let $\{P_i^e,Q_i^e\}$ be the core stable partition of $H_i^e\setminus \{h_i^e\}$ defined in the previous paragraph. 
	For every edge $e=xy\in E(G)$ we introduce  vertices $v_1^e,v_2^e,v_3^e$ and add edges $v_1^ev_2^e$, $v_1^ev_3^e$ and $v_2^ev_3^e$ of weight $\rho$.   We add  edges  $v_i^eu_i^x$, $v_i^eu_i^y$ of weight $5$ for every $i\in [3]$. 
	Furthermore, for every $e\in E(G)$ we add $3$ copies $(H_1^e,w_1^e),(H_2^e,w_2^e),(H_3^e,w_3^e)$ of the graph $(H,w_\rho)$ and let $h_i^e$ be the vertex $h$ in the copy $H_i^e$ for every $i\in [3]$. We $9$-neighborhood attach $(H_i^e,w_i^e)$ at $\{v_i^e\}$. \\

	Note that since $G$ is sub-cubic the resulting graph has degree $14$. This concludes the construction of $(\widehat{G},\widehat{w})$. For an illustration of the construction see Figure~\ref{fig:3csf:3col}. We now argue that $G$ is $3$-colorable if and only if $(\widehat{G},\widehat{w})$ admits a $k$-core stable partition.\\

	First assume that $G$ is $3$-colorable and let $c:V(G)\rightarrow [3]$ be a proper $3$-coloring of $G$. Recall that $\mathcal{P}_{(H_i^e,w_i^e)}$ is a core stable partition of $(H_i^e\setminus \{h_i^e\},w_i^e|_{V(H_i^e\setminus\{h_i^e\})})$ defined in Section~\ref{sec:auxiliaryGadget}. Consider the following partition $\mathcal{P}$ of $V(\widehat{G})$. For every vertex $x\in V(G)$ the partition $\mathcal{P}$ contains sets $\{u_{6+c(x)}^x\}$, $\{u_{3+c(x)}^x, u_1^x,u_2^x,u_3^x\}$ and $\{u_{3+i}^x,u_{6+i}^x\}$ for $i\in [3]$, $i\not=c(x)$. For every edge $e\in E(G)$  the partition $\mathcal{P}$ contains set $\{h_i^e,v_i^e\}$ and $\mathcal{P}$ contains $\mathcal{P}_{(H_i^e,w_i^e)}$ for every $i\in [3]$. For an illustration of the partition $\mathcal{P}$ see Figure~\ref{fig:3csf:3col}.
	
	We now argue that $\mathcal{P}$ is $k$-core stable. Note that it is sufficient to argue that there is no blocking coalition $X$ such that $|X|\leq k$ and $\widehat{G}(X)$ is connected. Since every vertex $u\in V(\widehat{G})$ has positive utility in $\mathcal{P}$ we can further assume that $|X|\geq 2$. We proceed by contradiction and assume that $X$ is  a blocking coalition such that $2\leq |X|\leq k$. Without loss of generality we can assume that $\widehat{G}[X]$ is connected. %First observe that by choice of edges of weight $\rho$ we can  assume that $X$ induces a connected graph in $\widehat{G}^+:=(V(\widehat{G}), E^+, \widehat{w}|_{E^+} )$ where $E^+:=\{e\in E(\widehat{G}):\widehat{w}(e)\geq 0\}$. 
	
	We now argue that for every vertex $u\in V(\widehat{G})$ that $u\notin X$. We first argue that for every $x\in V(G)$ the set $\{u_1^x,u_2^x,u_3^x\}\not\subset X$. For this note that  two vertices from each set $\{u_1^x,u_2^x,u_3^x\}$ have utility 10 in $\mathcal{P}$. Since by construction only at most one vertex from $\{u_4^x,u_5^x,u_6^x\}$ can be in $X$, the utility of every vertex in $\{u_1^x,u_2^x,u_3^x\}$ can be at most 15. Additionally, by construction only one vertex from $\{v_1^e,v_2^e,v_3^e\}$ for every edge $e\in E(G)$ can be contained in $X$. This implies that the utility of at most one vertex in $\{u_1^x,u_2^x,u_3^x\}$  can be 15 in $X$ and the utility of the other two is at most 10. Hence, only at most one vertex of the two vertices in $\{u_1^x,u_2^x,u_3^x\}$ of utility 10 in $\mathcal{P}$ can improve its utility by joining $X$ and therefore $\{u_1^x,u_2^x,u_3^x\}\not\subset X$. Now, observe that for every $x\in V(G)$ the vertex $u_{6+i}^x$ for $i\in [3]$, $i\not=c(x)$ cannot be contained in $X$ as it has maximum utility in $\mathcal{P}$. 
	This implies that  vertex  $u=u_{3+i}^x$ for $x\in V(G)$, $i\in [3]$, $i\not=c(x)$ cannot be in $X$ as it has utility $50$ in $\mathcal{P}$ and $u$ could have utility at most 40 in $X$ as $\{u_1^x,u_2^x,u_3^x\}\not\subset X$ and  $u_{6+i}^x\notin X$. Furthermore, if $u=u_{3+c(x)}^x\in X$ for some $x\in V(G)$, then $u_{6+c(x)}^x\in X$ as this is the only neighbor of $u$ with an edge to $u$ of positive weight, which is not contained in the same part of $\mathcal{P}$ as $u$. But then either $u_{1}^x\in X$, $u_{2}^x\in X$ or $u_{3}^x\in X$ to obtain a utility of $u_{3+c(x)}^x$ larger than $54$. But this is a contradiction as $u_{6+c(x)}^x u_i^x$ for $i\in [3]$ has weight $\rho$. Hence $u_{3+c(x)}^x\notin X$ for every $x\in V(G)$. This implies that also $u_{6+c(x)}^x\notin X$. Next observe that
	for every $x\in V(G)$ the vertices $u_{i}^x$ for $i\in [3]$, $i\not=c(x)$ have utility $10$ in $\mathcal{P}$ and only one positive weighted edge of weight $5$ which is not incident to any vertex which still potentially could be contained in $X$. Hence, $u_{i}^x\notin X$ for any $x\in V(G)$, $i\in [3]$, $i\not=c(x)$. Additionally, %observe that for every edge $e\in E(G)$ and $i\in [3]$ the set $B$ cannot be a subset of  $V(H_i^e)\setminus\{h_i^e\}$ since $\{P_i^e,Q_i^x\}$ is core stable. Furthermore, if $B$ contains any vertex $h$ from $V(H_i^e)\setminus \{h_i^e\}$ it has to also contain $h_i^e$ because of the assumption that $\widehat{G}^+[B]$ is connected. But then $B$ cannot contain $w_i^e$. Therefore, $B\subseteq V(H_i^e)$. But $h_i^e$ cannot form a coalition with only one vertex from $V(H_i^e)\setminus\{h_i^e\}$ because its utility would not improve. Furthermore, out of the two positive weight neighbors of $h_i^v$ the one contained in $Q_i^d$ cannot improve its utility by forming a coalition with $h_i^e$ and the other neighbor. 
	by Lemma~\ref{lem:auxiliaryNeighborhoodGadget} property~\ref{prop:gadgetP3} $X$ cannot contain any vertex from $V(H_i^e)$ for every $e\in E(G)$, $i\in [3]$.   
	
	In conclusion, we have argued that $X\subseteq \{v_i^e: e\in E(G),i\in [3]\}\cup \{u_{c(x)}^x:x\in V(G)\}$. Observe that by our assumption that $X$ is connected  and because $|X|\geq 2$ we get that there must be $e=xy\in E(G)$, $i\in [3]$ such that $v_i^e\in X$. Because $v_i^e$ has utility $9$ in $\mathcal{P}$ this implies that $u_{c(x)}^x\in X$ and $u_{c(y)}^y\in X$. But then by construction of $\widehat{G}$ we have that $c(x)=c(y)$ which is a contradiction to $c$ being a proper coloring. Hence, $(\widehat{G},\widehat{w})$ is a YES-instance of $k$-\textsc{CS}.
	\\

	Now assume that $\widehat{G}$ is $k$-core stable and let $\mathcal{P}$ be a $k$-core stable partition of $\widehat{G}$. Without loss of generality we assume that $\widehat{G}[P]$ is connected for every set $P\in \mathcal{P}$. The following  claim will be the key step for defining a proper coloring of $G$.
	\begin{claim}\label{claim:colorVertexGadget}
		For every vertex $x\in V(G)$ there is $j_x\in [3]$ such that  vertex $u_{j_x}^x$ has utility at most $4$ in $\mathcal{P}$. %while vertices $u_i^x$ for $i\in [3]$, $i\not=j_x$ have utility $10$ in $\mathcal{P}$.
	\end{claim}
	\begin{proof}
		Let $x\in V(G)$ be arbitrary. First note that for every edge $e\in E(G)$ and every $i\in [3]$ the set $\{v_i^e,h_i^e\}\in \mathcal{P}$ by Lemma~\ref{lem:auxiliaryNeighborhoodGadget}. Combining this with the assumption that $\widehat{G}[P]$ is connected for every $p\in \mathcal{P}$  we obtain that every $P\in \mathcal{P}$ containing a vertex $u_i^x$, $i\in [9]$ must be contained in $\{u_i^x:i\in [9]\}$. Hence, we let $\mathcal{P}^x\subseteq \mathcal{P}$ be the partition of $\{u_i^x:i\in [9]\}$ obtained by restricting $\mathcal{P}$ to sets containing vertices $u_i^x$, $i\in [9]$. Note that $\mathcal{P}^x$ must be $k$-core stable since any blocking coalition of $\mathcal{P}^x$ is a blocking coalition of $\mathcal{P}$.
		
		For $j\in [3]$ we let $B_j$ be the set $\{u_{3+j}^x,u_1^x,u_2^x,u_3^x\}$. First observe that if $B_j\subseteq P$ for some $P\in \mathcal{P}^x$ then $B_j=P$ since otherwise $\widehat{G}[P]$ has to contain an edge of weight $\rho$ implying that there is a blocking coalition of size $1\leq k$. In case there is a $j\in [3]$ such that $B_j\in \mathcal{P}^x$ we choose $j_x$ to be $j$ and observe that this choice satisfies the requirements stated in the claim. 
		
		Now consider the case that there is no $j\in [3]$ for which $B_j\in \mathcal{P}^x$. First note that in this case for every $j\in [3]$ the set $\{u_{3+j}^x,u_{6+j}^x\}$ must be in $\mathcal{P}^x$. If this is not the case then $\{u_{6+j}^x\}\in \mathcal{P}^x$  can have utility at most $40$ in $\mathcal{P}^x$ by assumption that $B_j\notin \mathcal{P}^x$. But then $\{u_{3+j}^x,u_{6+j}^x\}$ is a blocking coalition of $\mathcal{P}^x$ of size at most $k$ which is a contradiction.   This implies that for every $j\in [3]$ the utility of $u_j^x$ can be at most $0$ in $\mathcal{P}^x$. Hence we can choose $j_x\in [3]$ arbitrarily. 
	\end{proof}
	We can now define a $3$-coloring $c:V(G)\rightarrow [3]$ of $G$. For every vertex $x\in V(G)$ we let $j_x\in [3]$ as in Claim~\ref{claim:colorVertexGadget} and  set $c(x):=j_x$. First note that by Lemma~\ref{lem:auxiliaryNeighborhoodGadget} $v_i^e$ has utility $9$ in $\mathcal{P}$ for every edge $e\in E(G)$, $i\in [3]$.
	Additionally, if there is an edge $e=xy\in E(G)$ for which $j:=j_x=c(x)=c(y)=j_y$ then both $u_{j}^x$ and $u_{j}^y$ have utility at most $4$ in $\mathcal{P}$ by Claim~\ref{claim:colorVertexGadget}. But then $\{u_{j}^x,u_{j}^y,v_{j}^e\}$ is a blocking coalition of $\mathcal{P}$ of size at most $k$. Hence the coloring $c$ is proper and $G$ is a YES-instance of $3$-\textsc{Coloring}.
\end{appendixproof}

\section{Conclusion}

The general tenor of our results indicates that core stability is an
algorithmically highly intractable notion: even for very restricted input
structures, obtaining efficient algorithms seems out of reach; and even for the
few cases where positive fixed-parameter tractability results can be obtained,
complexity lower bounds still push the parameter dependence to prohibitive
levels. Despite the above, we believe that a promising avenue for further
research may be the further investigation of $k$-core stability. Even though we
have shown that parameterizing the problem by $k$ alone does not help, it would
be interesting to ask whether parameterizing at the same time by both $k$ and a
structural parameter (such as treewidth) could help us evade the lower bounds
that apply to each case individually.
Finally, investigating the parameterized complexity of core stability in other variants of hedonic games such as fractional hedonic games~\cite{AzizBBHOP19,HanakaIO23,FanelliMM21} is another promising direction.

%%
%% Bibliography
%%

%% Please use bibtex, 

\bibliography{ref}

\appendix

\end{document}